\documentclass[reqno]{amsart}

\usepackage{amssymb,bm,enumitem,mathrsfs,mleftright,parskip,soul}
\usepackage[margin = 1in]{geometry}
\usepackage[colorlinks = true]{hyperref}

\setlist[itemize]{leftmargin = *}
\setlist[enumerate]{leftmargin = *}

\allowdisplaybreaks

\newtheorem{Cor}{Corollary}
\newtheorem{Lem}{Lemma}
\newtheorem{Thm}{Theorem}
\newtheorem{Prop}{Proposition}

\theoremstyle{definition}

\newtheorem{Def}{Definition}

\newcommand{\C}{\mathbb{C}}
\newcommand{\F}[2]{\mathcal{F}_{#1}^{#2}}
\newcommand{\K}{\mathcal{K}}
\newcommand{\M}{\mathsf{M}}
\newcommand{\N}{\mathbb{N}}
\newcommand{\R}{\mathbb{R}}
\newcommand{\U}{\mathsf{U}}
\newcommand{\V}{\mathsf{V}}
\newcommand{\X}{\mathsf{X}}
\newcommand{\Y}{\mathsf{Y}}
\newcommand{\Z}{\mathsf{Z}}
\newcommand{\Br}[1]{\mleft( #1 \mright)}
\newcommand{\Cb}[1]{\func{\operatorname{C_{b}}}{#1}}
\newcommand{\Cc}[1]{\func{\operatorname{C_{c}}}{#1}}
\newcommand{\CC}[1]{\func{\operatorname{C}}{#1}}
\newcommand{\Cl}[2]{\overline{#1}^{#2}}
\newcommand{\Co}[1]{\func{\operatorname{C_{0}}}{#1}}
\newcommand{\df}{\stackrel{\textnormal{df}}{=}}
\newcommand{\ds}{\displaystyle}
\newcommand{\FT}[3]{\func{\F{#1}{#2}}{#3}}
\newcommand{\Id}{\operatorname{Id}}
\newcommand{\LL}[3]{\func{\mathsf{L}^{2}}{#1,#2,#3}}
\newcommand{\lt}{\mathsf{lt}}
\newcommand{\MO}{\widehat{P_{\hbar}}}
\newcommand{\PO}{\widehat{X}}
\newcommand{\Abs}[1]{\mleft| #1 \mright|}
\newcommand{\Adj}[1]{\func{\mathbb{L}}{#1}}
\newcommand{\Bdd}[1]{\func{\mathbb{B}}{#1}}
\newcommand{\Bra}[1]{\mleft\langle #1 \mright|}

\newcommand{\Dom}[1]{\func{\operatorname{Dom}}{#1}}
\newcommand{\Fin}[1]{\func{\operatorname{Fin}}{#1}}
\newcommand{\Int}[4]{\int_{#1} #2 ~ \d{\func{#3}{#4}}}
\newcommand{\Ket}[1]{\mleft| #1 \mright\rangle}
\newcommand{\Map}[4]{\mleft\{ \begin{matrix} #1 & \to & #2 \\ #3 & \mapsto & \ds #4 \end{matrix} \mright\}}
\newcommand{\Seq}[2]{\Br{#1}_{#2}}
\newcommand{\Set}[2]{\SSet{#1 ~ \middle| ~ #2}}
\newcommand{\Comm}[2]{\SqBr{#1,#2}}
\newcommand{\Comp}[1]{\func{\mathbb{K}}{#1}}
\newcommand{\func}[2]{#1 \Br{#2}}
\newcommand{\Func}[2]{\func{\Br{#1}}{#2}}
\newcommand{\FUNC}[2]{\func{\SqBr{#1}}{#2}}
\newcommand{\into}{\hookrightarrow}
\newcommand{\Norm}[1]{\mleft\| #1 \mright\|}
\newcommand{\Pair}[2]{\Br{#1,#2}}
\newcommand{\Proj}[2]{\operatorname{Proj}_{#1,#2}}
\newcommand{\Quad}[4]{\Br{#1,#2,#3,#4}}
\newcommand{\Span}[1]{\func{\operatorname{Span}}{#1}}
\newcommand{\SqBr}[1]{\mleft[ #1 \mright]}
\newcommand{\SSet}[1]{\mleft\{ #1 \mright\}}
\newcommand{\Supp}[1]{\func{\operatorname{Supp}}{#1}}
\newcommand{\Trip}[3]{\Br{#1,#2,#3}}
\newcommand{\Cstar}[1]{\func{C^{\ast}}{#1}}
\newcommand{\Inner}[2]{\mleft\langle #1 \middle| #2 \mright\rangle}
\newcommand{\Range}[1]{\func{\operatorname{Range}}{#1}}
\newcommand{\betArg}[2]{\func{\beta_{#1}}{#2}}
\newcommand{\alphArg}[2]{\func{\alph{#1}}{#2}}
\newcommand{\SqInner}[2]{\SqBr{#1 \middle| #2}}
\newcommand{\Unitary}[1]{\func{\mathbb{U}}{#1}}

\renewcommand{\d}{\mathrm{d}}
\renewcommand{\H}{\mathcal{H}}
\renewcommand{\L}[2]{\func{L^{#1}}{#2}}
\renewcommand{\u}{\operatorname{u}}
\renewcommand{\Im}[2]{#1 \SqBr{#2}}
\renewcommand{\alph}[1]{\alpha_{#1}}

\title[The Covariant Stone-von Neumann Theorem]
      {The Covariant Stone-von Neumann Theorem for Actions of Abelian Groups on $ C^{\ast} $-Algebras of Compact Operators}

\author{Leonard Huang}
\address{Leonard Huang, Department of Mathematics \& Statistics, University of Nevada, Reno, 1664 N. Virginia Street, Reno, NV 89557}
\email{Leonard.Huang@unr.edu}

\author{Lara Ismert}
\address{Lara Ismert, Department of Mathematics, University of Nebraska-Lincoln, 1400 R Street, Lincoln, NE 68588}
\email{lara.ismert@huskers.unl.edu}


\begin{document}

\maketitle


\begin{abstract}
In this paper, we formulate and prove a version of the Stone-von Neumann Theorem for every $ C^{\ast} $-dynamical system of the form $ \Trip{G}{\Comp{\H}}{\alpha} $, where $ G $ is a locally compact Hausdorff abelian group and $ \H $ is a Hilbert space. The novelty of our work stems from our representation of the Weyl Commutation Relation on Hilbert $ \Comp{\H} $-modules instead of just Hilbert spaces, and our introduction of two additional commutation relations, which are necessary to obtain a uniqueness theorem. Along the way, we apply one of our basic results on Hilbert $ C^{\ast} $-modules to significantly shorten the length of Iain Raeburn's well-known proof of Takai-Takesaki Duality.
\end{abstract}



\section{Introduction}


One of the most famous mathematical results in quantum mechanics is the Stone-von Neumann Theorem. Informally, the theorem establishes the physical equivalence of Werner Heisenberg's matrix mechanics and Erwin Schr\"{o}dinger's wave mechanics, which was seen by Heisenberg to be an outstanding problem in the early days of quantum mechanics (\cite{He}). The theorem was an attempt to prove that any pair $ \Pair{A}{B} $ of self-adjoint unbounded operators on a Hilbert space $ \H $ that satisfies the Heisenberg Commutation Relation on a common dense invariant subset $ D $ of their domain, i.e.,
$$
\Comm{A|_{D}}{B|_{D}} = i \hbar \cdot \Id_{D},
$$
is unitarily equivalent to a direct sum of copies of $ \Pair{\PO}{\MO} $, which are self-adjoint unbounded operators on $ \L{2}{\R} $ defined as follows:
\begin{alignat*}{2}
&  \Dom{\PO} = \Set{f \in \L{2}{\R}}{\Int{\R}{\Abs{x \func{f}{x}}^{2}}{\mu}{x}}, \qquad
&& \forall f \in \Dom{\PO}: \quad
   \func{\PO}{f} \df \Map{\R}{\C}{x}{x \func{f}{x}}; \\
&  \Dom{\MO} = \func{W^{1,2}}{\R},
&& \forall f \in \Dom{\MO}: \quad
   \func{\MO}{f} \df - i \hbar \cdot f'.
\end{alignat*}
We recall that $ \func{W^{1,2}}{\R} $ denotes the space of weakly-differentiable square-integrable functions on $ \R $ whose weak derivative is also square-integrable. This statement about unbounded operators is not true in general, where complications arise from domain issues --- a well-known counterexample involving the Hilbert space of $ L^{2} $-functions on a two-sheeted Riemann surface is given in \cite{Ne}.

The Stone-von Neumann Theorem was first given a rigorous formulation by Marshall Stone in 1930 (\cite{St}), and it was this formulation that
John von Neumann proved in 1931 (\cite{vN}). The exponentiated form of the Heisenberg Commutation Relation, called the \emph{Weyl Commutation Relation}, is investigated in these papers, because it involves only one-parameter unitary groups. More precisely, a pair $ \Pair{R}{S} $ of strongly-continuous one-parameter unitary groups on a Hilbert space $ \H $ satisfies the Weyl Commutation Relation if and only if
$$
\forall x,y \in \R: \qquad
\func{S}{y} \func{R}{x} = e^{- i \hbar x y} \cdot \func{R}{x} \func{S}{y}.
$$
von Neumann proved that any such pair is unitarily equivalent to a direct sum of copies of $ \Pair{\U}{\V} $, where $ \U $ and $ \V $ are strongly-continuous one-parameter unitary groups on $ \L{2}{\R} $ defined by
$$
\forall x,y \in \R, ~ \forall f \in \L{2}{\R}: \qquad
\FUNC{\func{\U}{x}}{f} = \func{f}{\bullet + \hbar x}
\qquad \text{and} \qquad
\FUNC{\func{\V}{y}}{f} = e^{i y \bullet} f.
$$
Basically, $ \U $ acts by translations, and $ \V $ acts by phase modulations.

The statement of the Stone-von Neumann Theorem has undergone major revisions in the decades since its initial formulation. George Mackey appears to have been the first to recognize its generalization to second-countable locally compact Hausdorff abelian groups, in \cite{Mac}. Nowadays, his generalization is treated as part of his theory of induced representations of locally compact Hausdorff groups, and is generally considered the standard modern formulation of the Stone-von Neumann Theorem, which we now state.



\begin{Thm} \label{The Stone-von Neumann Theorem}
Let $ G $ be a locally compact Hausdorff abelian group. If $ R $ and $ S $ are strongly-continuous unitary representations of $ G $ and $ \widehat{G} $, respectively, on a Hilbert space $ \H $ that satisfy the Weyl Commutation Relation, i.e.,
$$
\forall x \in G, ~ \forall \gamma \in \widehat{G}: \qquad
\func{S}{\varphi} \func{R}{x} = \func{\varphi}{x} \cdot \func{R}{x} \func{S}{\varphi},
$$
then $ \Trip{\H}{R}{S} $ must be unitarily equivalent to a direct sum of copies of $ \Trip{\L{2}{G}}{\U_{G}}{\V_{G}} $, where $ \U_{G} $ denotes the unitary representation of $ G $ on $ \L{2}{G} $ by left translations, and $ \V_{G} $ denotes the unitary representation of $ \widehat{G} $ on $ \L{2}{G} $ by phase modulations, i.e.,
$$
\FUNC{\func{\U_{G}}{x}}{f} \df \func{f}{x^{- 1} \bullet}
\qquad \text{and} \qquad
\FUNC{\func{\V_{G}}{\varphi}}{f} \df \varphi f
$$
for all $ x \in G $, $ \varphi \in \widehat{G} $, and $ f \in \L{2}{G} $.
\end{Thm}


The work of Marc Rieffel in \cite{Ri1,Ri2} has revealed that \autoref{The Stone-von Neumann Theorem} is actually a statement about the Morita equivalence of the $ C^{\ast} $-algebras $ \C $ and $ \Cstar{G,\Co{G},\lt} $, where $ \lt $ denotes the strongly-continuous action of $ G $ on $ \Co{G} $ by left translations. The theorem thus acquires a more algebraic flavor. This Morita equivalence is a special case of a more general result known as \emph{Green's Imprimitivity Theorem}, which we actually need to prove our covariant generalization of \autoref{The Stone-von Neumann Theorem}.

Several generalizations of the Stone-von Neumann Theorem can be found in the literature. For example, \cite{CaMoSt} extends the theorem to measurable unitary representations of $ G $ and $ \widehat{G} $ on a Hilbert space, and \cite{Pa} extends the theorem to Hecke pairs using the machinery of non-abelian duality. Although these generalizations are non-trivial and interesting, their use of only Hilbert-space representations is a common limiting feature.

In this paper, we provide not another incremental generalization of the Stone-von Neumann Theorem, but a complete paradigm shift that significantly augments the theorem's range of applicability. By leaving the realm of Hilbert spaces and working with representations on Hilbert $ C^{\ast} $-modules, we show that the Stone-von Neumann Theorem is not really about representations of locally compact Hausdorff abelian groups on Hilbert spaces, but is really about representations of $ C^{\ast} $-dynamical systems on Hilbert $ C^{\ast} $-modules. More precisely, for every $ C^{\ast} $-dynamical system of the form $ \Trip{G}{\Comp{\H}}{\alpha} $, where $ G $ is a locally compact Hausdorff abelian group and $ \H $ is a Hilbert space, our covariant generalization classifies up to unitary equivalence all quadruples $ \Quad{\X}{\rho}{R}{S} $ with the following properties:
\begin{itemize}
\item
$ \X $ is a non-trivial Hilbert $ \Comp{\H} $-module.

\item
$ R $ and $ S $ are strongly-continuous unitary representations of $ G $ and $ \widehat{G} $, respectively, on $ \X $ that satisfy the Weyl Commutation Relation.

\item
$ \rho $ is a non-degenerate $ \ast $-representation of $ \Comp{\H} $ on $ \X $ that obeys the following commutation relations:
$$
\func{R}{x} \func{\rho}{a} = \func{\rho}{\alphArg{x}{a}} \func{R}{x}
\qquad \text{and} \qquad
\func{S}{\varphi} \func{\rho}{a} = \func{\rho}{a} \func{S}{\varphi}
$$
for all $ x \in G $, $ \varphi \in \widehat{G} $, and $ a \in A $. These relations are also called \emph{covariance relations}.
\end{itemize}

Using results on non-abelian duality, one could very well generalize our covariant version of the Stone-von Neumann Theorem to non-abelian $ C^{\ast} $-dynamical systems, or even quantum-group dynamical systems, but such an undertaking would take us too far afield, so we content ourselves with presenting only the abelian case, which we feel is already a significant advance. Further generalizations will be explored in a sequel.

This paper is organized as follows:
\begin{itemize}
\item
Section 2 is a short preliminary section that recalls some concepts and results about $ C^{\ast} $-crossed products that we need. In particular, we show how to associate a Hilbert $ C^{\ast} $-module to a $ C^{\ast} $-dynamical system in a canonical way. This Hilbert $ C^{\ast} $-module is featured in Green's Imprimitivity Theorem, and is crucial to a formulation of our covariant generalization of the Stone-von Neumann Theorem.

\item
Section 3 introduces Heisenberg modular representations and the Schr\"odinger modular representation of an abelian $ C^{\ast} $-dynamical system $ \Trip{G}{A}{\alpha} $. These concepts allow an efficient formulation of our covariant generalization of the Stone-von Neumann Theorem. We construct an injective map from the class of all Heisenberg modular representations of $ \Trip{G}{A}{\alpha} $ to the class of all covariant modular representations of $ \Trip{G}{\Co{G,A}}{\lt \otimes \alpha} $, which is a $ C^{\ast} $-dynamical system that plays a pivotal role in Iain Raeburn's proof of Takai-Takesaki Duality. A basic result in this section allows us to significantly shorten his proof.

\item
Section 4 provides an overview of the properties of Hilbert $ \Comp{\H} $-modules that have been established in \cite{BaGu1,BaGu2}. Hilbert $ \Comp{\H} $-modules obviously generalize Hilbert spaces, yet they behave very much like Hilbert spaces, which makes them very desirable to work with.

\item
Section 5 contains our main result: the Covariant Stone-von Neumann Theorem. We give a statement of Green's Imprimitivity Theorem, and explain its relevance to the main result.

\item
Section 6 poses some open questions that this paper was unable to answer. It also suggests new avenues of research that would be of interest to both mathematicians and physicists.

\item
Finally, an appendix contains proofs of two results, stated in the main body of this paper, that would be considered folklore, but for which we were unable to locate adequate references.
\end{itemize}

We assume that the reader has a reasonable working knowledge of $ C^{\ast} $-algebras, $ C^{\ast} $-dynamical systems, and Hilbert $ C^{\ast} $-modules. Throughout this paper, we adopt the following notations and conventions:
\begin{itemize}
\item
$ \N $ denotes the set of positive integers, and for each $ n \in \N $, let $ \SqBr{n} \df \N_{\leq n} $.

\item
For a set $ I $, let $ \Fin{I} $ denote the set of finite subsets of $ I $.

\item
For a locally compact Hausdorff abelian group $ G $, let $ \widehat{G} $ denote its Pontryagin dual.

\item
For a locally compact Hausdorff space $ X $ and a normed vector space $ V $, let $ \diamond: \Co{X} \times V \to \Co{X,V} $ be defined by
$$
\forall f \in \Cc{X}, ~ \forall v \in V: \qquad
f \diamond v \df \Map{X}{V}{x}{\func{f}{x} \cdot v}.
$$
Note that $ \diamond $ takes $ \Cc{X} \times V $ to $ \Cc{X,V} $.

\item
For a Hilbert space $ \H $ and vectors $ v,w \in \H $, let $ \Ket{v} \Bra{w} $ denote the rank-one operator on $ \H $ defined by
$$
\forall x \in \H: \qquad
\Func{\Ket{v} \Bra{w}}{x} \df \Inner{w}{x}_{\H} \cdot v.
$$

\item
For a Hilbert space $ \H $ and a closed subspace $ \K $ of $ \H $, let $ \Proj{\H}{\K} $ denote the orthogonal projection of $ \H $ onto $ \K $.

\item
For a $ C^{\ast} $-algebra $ A $, let $ A^{\sim} $ denote its minimal unitization.

\item
For a $ C^{\ast} $-algebra $ A $ and $ a \in A $, let $ \func{\sigma_{A}}{a} $ denote the spectrum of $ a $.

\item
For a $ C^{\ast} $-algebra $ A $, let $ \leq_{A} $ denote the usual partial order on the cone $ A_{\geq} $ of positive elements of $ A $.

\item
For a $ C^{\ast} $-algebra $ A $, and Hilbert $ A $-modules $ \X $ and $ \Y $, the set of adjointable/compact/unitary operators from $ \X $ to $ \Y $ is denoted by $ \Adj{\X,\Y} $/$ \Comp{\X,\Y} $/$ \Unitary{\X,\Y} $. If $ \X = \Y $, then we write $ \Adj{\X} $/$ \Comp{\X} $/$ \Unitary{\X} $.

\item
For a $ C^{\ast} $-algebra $ A $ and a Hilbert $ C^{\ast} $-module $ \X $ (not necessarily over $ A $), a $ \ast $-representation of $ A $ on $ \X $ is a $ C^{\ast} $-homomorphism $ \rho: A \to \Adj{\X} $, which is then said to be non-degenerate if and only if
$$
\Cl{\Span{\Set{\FUNC{\func{\rho}{a}}{\zeta}}{a \in B ~ \text{and} ~ \zeta \in \X}}}{\X} = \X.
$$

\item
For a locally compact Hausdorff group $ G $ and a Hilbert $ C^{\ast} $-module $ \X $, a unitary representation of $ G $ on $ \X $ is a group homomorphism $ R $ from $ G $ to the group $ \Unitary{\X} $ of unitary adjointable operators on $ \X $, which is then said to be strongly continuous if and only if the map
$$
\Map{G}{\X}{x}{\FUNC{\func{R}{x}}{\zeta}}
$$
is continuous for each $ \zeta \in \X $.
\end{itemize}



\section{Preliminaries}


As $ C^{\ast} $-crossed products will be used extensively in this paper, let us recall some concepts in this area.

Throughout this section, we shall fix an arbitrary $ C^{\ast} $-dynamical system $ \Trip{G}{A}{\alpha} $, with $ G $ not assumed to be abelian, and we shall fix a Haar measure $ \mu $ on $ G $.

Recall that the $ \C $-vector space $ \Cc{G,A} $ can be given a convolution $ \star_{G,A,\alpha} $ and an involution $ ^{\ast_{G,A,\alpha}} $ by
\begin{align*}
\forall f,g \in \Cc{G,A}: \qquad
f \star_{G,A,\alpha} g & \df \Map{G}{A}{x}{\Int{G}{\func{f}{y} \alphArg{y}{\func{g}{y^{- 1} x}}}{\mu}{y}}; \\
f^{\ast_{G,A,\alpha}}  & \df \Map{G}{A}{x}{\func{\Delta_{G}}{x^{- 1}} \cdot \alphArg{x}{\func{f}{x^{- 1}}^{\ast}}},
\end{align*}
where $ \Delta_{G} $ denotes the modular function of $ G $.



\begin{Def} \label{Covariant Modular Representation}
A \emph{$ \Trip{G}{A}{\alpha} $-covariant modular representation} is a triple $ \Trip{\X}{\rho}{R} $ with the following properties:
\begin{enumerate}
\item
$ \X $ is a Hilbert $ C^{\ast} $-module (not necessarily over $ A $).

\item
$ \rho $ is a non-degenerate $ \ast $-representation of $ A $ on $ \X $.

\item
$ R $ is a strongly-continuous unitary representation of $ G $ on $ \X $.

\item
$ \func{R}{x} \func{\rho}{a} = \func{\rho}{\alphArg{x}{a}} \func{R}{x} $ for all $ x \in G $ and $ a \in A $.
\end{enumerate}
\end{Def}


Covariant modular representations are used in the construction of $ C^{\ast} $-crossed products. Given a $ \Trip{G}{A}{\alpha} $-covariant modular representation $ \Trip{\X}{\rho}{R} $, we can define an algebraic $ \ast $-homomorphism $ \Pi_{\X,\rho,R} $, called the \emph{integrated form} of $ \Trip{\X}{\rho}{R} $, from the convolution $ \ast $-algebra $ \Trip{\Cc{G,A}}{\star_{G,A,\alpha}}{^{\ast_{G,A,\alpha}}} $ to $ \Adj{\X} $ by
$$
\forall f \in \Cc{G,A}: \qquad
\func{\Pi_{\X,\rho,R}}{f} \df \Int{G}{\func{\rho}{\func{f}{x}} \func{R}{x}}{\mu}{x}.
$$
The \emph{full crossed product} $ \Cstar{G,A,\alpha} $ is defined as the $ C^{\ast} $-algebraic completion of $ \Trip{\Cc{G,A}}{\star_{G,A,\alpha}}{^{\ast_{G,A,\alpha}}} $ with respect to the universal norm $ \Norm{\cdot}_{\Trip{G}{A}{\alpha},\u} $ given by
$$
    \Norm{f}_{\Trip{G}{A}{\alpha},\u}
\df \func{\sup}
         {
         \Set{\Norm{\func{\Pi_{\X,\rho,R}}{f}}_{\Adj{\X}}}
             {\Trip{\X}{\rho}{R} ~ \text{is a} ~ \Trip{G}{A}{\alpha} \text{-covariant modular representation}}
         }
$$
for all $ f \in \Cc{G,A} $. This norm is well-defined as it is dominated by the $ L^{1} $-norm on $ \Cc{G,A} $.

We let $ \eta_{\Trip{G}{A}{\alpha}} $ denote the canonical dense linear embedding of $ \Cc{G,A} $ into $ \Cstar{G,A,\alpha} $, and if $ A = \C $, in which case $ \alpha $ is necessarily trivial, we simply write $ \eta_{G} $.

For a $ \Trip{G}{A}{\alpha} $-covariant modular representation $ \Trip{\X}{\rho}{R} $, we denote by $ \overline{\Pi}_{\X,\rho,R} $ the extension of $ \Pi_{\X,\rho,\R} $ to a $ C^{\ast} $-homomorphism from $ \Cstar{G,A,\alpha} $ to $ \Adj{\X} $.



\begin{Lem} \label{Recovering a Covariant Modular Representation from Its Integrated Form}
Let $ x \in G $ and $ a \in A $. Then we can find nets $ \Seq{f_{i}}{i \in I} $ and $ \Seq{g_{j}}{j \in J} $ in $ \Cc{G,A} $ such that for any $ \Trip{G}{A}{\alpha} $-covariant modular representation $ \Trip{\X}{\rho}{R} $, the associated nets $ \Seq{\func{\Pi_{\X,\rho,R}}{f_{i}}}{i \in I} $ and $ \Seq{\func{\Pi_{\X,\rho,R}}{g_{j}}}{j \in J} $ strongly converge in $ \Adj{\X} $, respectively, to $ \func{R}{x} $ and $ \func{\rho}{a} $.
\end{Lem}


A proof of this lemma will be provided in the appendix.

To $ \Trip{G}{A}{\alpha} $, one can associate a special Hilbert $ A $-module, denoted by $ \LL{G}{A}{\alpha} $, in a canonical manner. Observe that $ \Cc{G,A} $ is a pre-Hilbert $ A $-module, whose right $ A $-action $ \bullet $ and $ A $-valued pre-inner product $ \SqInner{\cdot}{\cdot}: \Cc{G,A} \times \Cc{G,A} \to A $ are defined as follows:
\begin{itemize}
\item
$ \phi \bullet a \df \Map{G}{A}{x}{\func{\phi}{x} \alphArg{x}{a}} $ for all $ a \in A $ and $ \phi \in \Cc{G,A} $.

\item
$ \ds \SqInner{\phi}{\psi} \df \Int{G}{\alphArg{x^{- 1}}{\func{\phi}{x}^{\ast} \func{\psi}{x}}}{\mu}{x} $ for all $ \phi,\psi \in \Cc{G,A} $.
\end{itemize}
Define $ \LL{G}{A}{\alpha} $ to be the Hilbert $ A $-module obtained by completing $ \Cc{G,A} $ with respect to the norm induced by $ \SqInner{\cdot}{\cdot} $. Let $ q_{\Trip{G}{A}{\alpha}}: \Cc{G,A} \into \LL{G}{A}{\alpha} $ denote the canonical dense linear embedding, and if no confusion can arise, we will omit the subscript and simply write $ q $.

This Hilbert $ A $-module is the linchpin of our formulation of the covariant Stone-von Neumann Theorem.

We will use $ q_{\Trip{G}{A}{\alpha}} $ when defining operators on $ \LL{G}{A}{\alpha} $ as a way of emphasizing that unless $ A = \C $, the elements of $ \LL{G}{A}{\alpha} $ are generally not functions of any sort from $ G $ to $ A $.



\section{Modular Representations}


Throughout this section, we shall fix an arbitrary $ C^{\ast} $-dynamical system $ \Trip{G}{A}{\alpha} $ with $ G $ abelian. We shall also fix a Haar measure $ \mu $ on $ G $.



\begin{Def} \label{Heisenberg Modular Representation}
A \emph{$ \Trip{G}{A}{\alpha} $-Heisenberg modular representation} is a quadruple $ \Quad{\X}{\rho}{R}{S} $ with the following properties:
\begin{enumerate}
\item
$ \X $ is a full Hilbert $ A $-module.

\item
$ \rho $ is a non-degenerate $ \ast $-representation of $ A $ on $ \X $.

\item
$ R $ is a strongly continuous representation of $ G $ on $ \X $.

\item
$ S $ is a strongly continuous representation of $ \widehat{G} $ on $ \X $.

\item
$ \func{S}{\varphi} \func{R}{x} = \func{\varphi}{x} \cdot \func{R}{x} \func{S}{\varphi} $ for all $ x \in G $ and $ \varphi \in \widehat{G} $. \\
Hence, $ \Pair{R}{S} $ satisfies the Weyl Commutation Relation for $ G $ on $ \X $.

\item
$ \func{R}{x} \func{\rho}{a} = \func{\rho}{\alphArg{x}{a}} \func{R}{x} $ for all $ x \in G $ and $ a \in A $. \\
Hence, $ \Trip{\X}{\rho}{R} $ is a $ \Trip{G}{A}{\alpha} $-covariant modular representation.

\item
$ \func{S}{\varphi} \func{\rho}{a} = \func{\rho}{a} \func{S}{\varphi} $ for all $ \varphi \in \widehat{G} $ and $ a \in A $. \\
Hence, $ \Trip{\X}{\rho}{S} $ is a $ \Trip{\widehat{G}}{A}{\iota} $-covariant modular representation, with $ \iota $ denoting the trivial action of $ \widehat{G} $ on $ A $.
\end{enumerate}
\end{Def}


With the aim of producing an example of a $ \Trip{G}{A}{\alpha} $-Heisenberg modular representation, let us first equip $ \LL{G}{A}{\alpha} $, as defined in the Preliminaries, with the following structural data:
\begin{itemize}
\item
A $ \ast $-representation $ \M^{\Trip{G}{A}{\alpha}} $ of $ A $ on $ \LL{G}{A}{\alpha} $ such that for all $ a \in A $ and $ \phi \in \Cc{G,A} $,
$$
\FUNC{\func{\M^{\Trip{G}{A}{\alpha}}}{a}}{\func{q}{\phi}} = \func{q}{\Map{G}{A}{x}{a \func{\phi}{x}}}.
$$

\item
A unitary representation $ \U^{\Trip{G}{A}{\alpha}} $ of $ G $ on $ \LL{G}{A}{\alpha} $ such that for all $ x \in G $ and $ \phi \in \Cc{G,A} $,
$$
\FUNC{\func{\U^{\Trip{G}{A}{\alpha}}}{x}}{\func{q}{\phi}} = \func{q}{\Map{G}{A}{y}{\alphArg{x}{\func{\phi}{x^{- 1} y}}}}.
$$

\item
A unitary representation $ \V^{\Trip{G}{A}{\alpha}} $ of $ G $ on $ \LL{G}{A}{\alpha} $ such that for all $ \varphi \in \widehat{G} $ and $ \phi \in \Cc{G,A} $,
$$
\FUNC{\func{\V^{\Trip{G}{A}{\alpha}}}{\varphi}}{\func{q}{\phi}} = \func{q}{\varphi \cdot \phi}.
$$
\end{itemize}
Proving that these representations are well-defined is a routine exercise. We refer the reader to Chapter 4 of \cite{Wi} for details. We will omit supscripts and simply write $ \M $, $ \U $, and $ \V $ if no confusion arises from doing so.



\begin{Def} \label{Schrödinger Modular Representation}
$ \Quad{\LL{G}{A}{\alpha}}{\M}{\U}{\V} $ is called the \emph{$ \Trip{G}{A}{\alpha} $-Schr\"{o}dinger modular representation}.
\end{Def}



\begin{Prop} \label{The Schrödinger Modular Representation Is a Heisenberg Modular Representation}
$ \Quad{\LL{G}{A}{\alpha}}{\M}{\U}{\V} $ is a $ \Trip{G}{A}{\alpha} $-Heisenberg modular representation.
\end{Prop}

\begin{proof}
We will verify the various axioms in \autoref{Heisenberg Modular Representation}.

\ul{The fullness of $ \LL{G}{A}{\alpha} $ as a Hilbert $ A $-module}

Fix an $ a \in A $. Using Urysohn's Lemma, find a $ \varphi \in \Cc{G,\R_{\geq 0}} $ such that $ \ds \Int{G}{\func{\varphi}{x}}{\mu}{x} = 1 $. Let $ \Seq{e_{\lambda}}{\lambda \in \Lambda} $ be an approximate identity for $ A $. Define $ \phi \in \Cc{G,A} $ and a net $ \Seq{\psi_{\lambda}}{\lambda \in \Lambda} $ in $ \Cc{G,A} $ by
$$
\phi \df \Map{G}{A}{x}{\sqrt{\func{\varphi}{x}} \cdot \alphArg{x}{a}^{\ast}}, \qquad
\forall \lambda \in \Lambda: \qquad
\psi_{\lambda} \df \Map{G}{A}{x}{\sqrt{\func{\varphi}{x}} \cdot \alphArg{x}{e_{\lambda}}}.
$$
Then for all $ \lambda \in \Lambda $,
\begin{align*}
    \Inner{\func{q}{\phi}}{\func{q}{\psi_{\lambda}}}_{\LL{G}{A}{\alpha}}
& = \Int{G}{\alphArg{x^{- 1}}{\func{\phi}{x}^{\ast} \func{\psi_{\lambda}}{x}}}{\mu}{x} \\
& = \Int{G}
        {
        \alphArg{x^{- 1}}{\SqBr{\sqrt{\func{\varphi}{x}} \cdot \alphArg{x}{a}} \SqBr{\sqrt{\func{\varphi}{x}} \cdot \alphArg{x}{e_{\lambda}}}}
        }
        {\mu}{x} \\
& = \Int{G}{\alphArg{x^{- 1}}{\func{\varphi}{x} \cdot \alphArg{x}{a e_{\lambda}}}}{\mu}{x} \\
& = \Int{G}{\func{\varphi}{x} \cdot a e_{\lambda}}{\mu}{x} \\
& = a e_{\lambda}.
\end{align*}
Hence, $ \ds \lim_{\lambda \in \Lambda} \Inner{\func{q}{\phi}}{\func{q}{\psi_{\lambda}}}_{\LL{G}{A}{\alpha}} = \lim_{\lambda \in \Lambda} a e_{\lambda} = a $. As $ a \in A $ is arbitrary, we get
$$
\Cl{\Span{\Inner{\LL{G}{A}{\alpha}}{\LL{G}{A}{\alpha}}_{\LL{G}{A}{\alpha}}}}{A} = A,
$$
which proves that $ \LL{G}{A}{\alpha} $ is a full Hilbert $ A $-module.

\ul{The non-degeneracy of $ \M $}

Fix $ \phi \in \Cc{G,A} $. As $ \Range{\phi} \subseteq \Im{\phi}{\Supp{\phi}} \cup \SSet{0_{A}} $, and as $ \Supp{\phi} $ is a compact subset of $ G $, we see that $ \Range{\phi} $ is contained in a compact subset of $ A $. Compact subsets of metric spaces are separable, and subsets of separable subsets of metric spaces are separable, so in particular, $ \Range{\phi} $ is a separable subset of $ A $. Let $ D $ be a countable dense subset of $ \Range{\phi} $. If $ B $ denotes the $ C^{\ast} $-subalgebra of $ A $ generated by $ \Range{\phi} $, then $ B $ is also the $ C^{\ast} $-subalgebra of $ A $ generated by $ D $. Hence, $ B $ is a separable $ C^{\ast} $-algebra, which means that it possesses a sequential approximate identity $ \Seq{e_{n}}{n \in \N} $ norm-bounded by $ 1 $. Now, for all $ n \in \N $,
\begin{align*}
       \Norm{\func{q}{\phi} - \FUNC{\func{\M}{e_{n}}}{\func{q}{\phi}}}_{\LL{G}{A}{\alpha}}
& =    \Norm{\func{q}{\Map{G}{A}{x}{\func{\phi}{x} - e_{n} \func{\phi}{x}}}}_{\LL{G}{A}{\alpha}} \\
& =    \Norm{
            \Int{G}
                {\alphArg{x^{- 1}}{\SqBr{\func{\phi}{x} - e_{n} \func{\phi}{x}}^{\ast} \SqBr{\func{\phi}{x} - e_{n} \func{\phi}{x}}}}
                {\mu}{x}
            }_{A}^{\frac{1}{2}} \\
& \leq \SqBr{
            \Int{G}
                {
                \Norm{\alphArg{x^{- 1}}{\SqBr{\func{\phi}{x} - e_{n} \func{\phi}{x}}^{\ast} \SqBr{\func{\phi}{x} - e_{n} \func{\phi}{x}}}}_{A}
                }
                {\mu}{x}
            }^{\frac{1}{2}} \\
& =    \SqBr{\Int{G}{\Norm{\SqBr{\func{\phi}{x} - e_{n} \func{\phi}{x}}^{\ast} \SqBr{\func{\phi}{x} - e_{n} \func{\phi}{x}}}_{A}}{\mu}{x}}
       ^{\frac{1}{2}} \\
& =    \SqBr{\Int{G}{\Norm{\func{\phi}{x} - e_{n} \func{\phi}{x}}_{A}^{2}}{\mu}{x}}^{\frac{1}{2}}.
\end{align*}
Next, notice for all $ n \in \N $ and $ x \in G $ that
\begin{align*}
       \Norm{\func{\phi}{x} - e_{n} \func{\phi}{x}}_{A}^{2}
& \leq \SqBr{\Norm{\func{\phi}{x}}_{A} + \Norm{e_{n} \func{\phi}{x}}_{A}}^{2} \\
& \leq \SqBr{\Norm{\func{\phi}{x}}_{A} + \Norm{e_{n}}_{A} \Norm{\func{\phi}{x}}_{A}}^{2} \\
& \leq \SqBr{\Norm{\func{\phi}{x}}_{A} + \Norm{\func{\phi}{x}}_{A}}^{2} \qquad \Br{\text{As $ \Norm{e_{n}}_{A} \leq 1 $.}} \\
& =    4 \Norm{\func{\phi}{x}}_{A}^{2}.
\end{align*}
Hence, $ \Seq{\Map{G}{\R_{\geq 0}}{x}{\Norm{\func{\phi}{x} - e_{n} \func{\phi}{x}}_{A}^{2}}}{n \in \N} $ is dominated by the integrable function $ \Map{G}{\R_{\geq 0}}{x}{4 \Norm{\func{\phi}{x}}_{A}^{2}} $, and as it converges pointwise to $ 0_{G \to \R_{\geq 0}} $, the Lebesgue Dominated Convergence Theorem yields
$$
\lim_{n \to \infty} \Norm{\func{q}{\phi} - \FUNC{\func{\M}{e_{n}}}{\func{q}{\phi}}}_{\LL{G}{A}{\alpha}} = 0.
$$
Finally, an $ \dfrac{\epsilon}{3} $-argument shows that for any $ \Phi \in \LL{G}{A}{\alpha} $ and any $ \epsilon > 0 $, there exists an $ a \in A $ such that $ \Norm{\Phi - \FUNC{\func{\M}{a}}{\Phi}}_{\LL{G}{A}{\alpha}} < \epsilon $. Therefore, $ \M $ is non-degenerate.

\ul{The strong continuity of $ \U $}

Fix $ \phi \in \Cc{G,A} $ and $ \epsilon > 0 $. We will show that there exists an open neighborhood $ W $ of $ e_{G} $ in $ G $ such that
$$
\forall x \in W: \qquad
\Norm{\func{q}{\phi} - \FUNC{\func{\U}{x}}{\func{q}{\phi}}}_{\LL{G}{A}{\alpha}} < \epsilon.
$$
Observe for all $ x \in G $ that
\begin{align*}
       \Norm{\func{q}{\phi} - \FUNC{\func{\U}{x}}{\func{q}{\phi}}}_{\LL{G}{A}{\alpha}}
& =    \Norm{\func{q}{\Map{G}{A}{y}{\func{\phi}{y} - \alphArg{x}{\func{\phi}{x^{- 1} y}}}}}_{\LL{G}{A}{\alpha}} \\
& =    \Norm{
            \Int{G}
                {
                \alphArg{y^{- 1}}
                        {
                        \SqBr{\func{\phi}{y} - \alphArg{x}{\func{\phi}{x^{- 1} y}}}^{\ast}
                        \SqBr{\func{\phi}{y} - \alphArg{x}{\func{\phi}{x^{- 1} y}}}
                        }
                }
                {\mu}{y}
            }_{A}^{\frac{1}{2}} \\
& \leq \SqBr{
            \Int{G}
                {
                \Norm{
                     \alphArg{y^{- 1}}
                             {
                             \SqBr{\func{\phi}{y} - \alphArg{x}{\func{\phi}{x^{- 1} y}}}^{\ast}
                             \SqBr{\func{\phi}{y} - \alphArg{x}{\func{\phi}{x^{- 1} y}}}
                             }
                     }_{A}
                }
                {\mu}{y}
            }^{\frac{1}{2}} \\
& =    \SqBr{
            \Int{G}
                {
                \Norm{
                     \SqBr{\func{\phi}{y} - \alphArg{x}{\func{\phi}{x^{- 1} y}}}^{\ast}
                     \SqBr{\func{\phi}{y} - \alphArg{x}{\func{\phi}{x^{- 1} y}}}
                     }_{A}
                }
                {\mu}{y}
            }^{\frac{1}{2}} \\
& =    \SqBr{\Int{G}{\Norm{\func{\phi}{y} - \alphArg{x}{\func{\phi}{x^{- 1} y}}}_{A}^{2}}{\mu}{y}}^{\frac{1}{2}}.
\end{align*}
Let $ K $ be a compact neighborhood of $ e_{G} $ in $ G $, and let $ x \in K $. Then
$$
\forall y \in G \setminus K \Supp{\phi}: \qquad
\Norm{\func{\phi}{y} - \alphArg{x}{\func{\phi}{x^{- 1} y}}}_{A} = 0
$$
for the following reasons:
\begin{itemize}
\item
$ \Supp{\phi} \subseteq K \Supp{\phi} $, so $ \func{\phi}{y} = 0_{A} $ for all $ y \in G \setminus K \Supp{\phi} $.

\item
For all $ y \in G \setminus K \Supp{\phi} $, we have $ x^{- 1} y \in G \setminus \Supp{\phi} $, so $ \func{\phi}{x^{- 1} y} = 0_{A} $; if $ x^{- 1} y \in \Supp{\phi} $, then $ y \in x \Supp{\phi} \subseteq K \Supp{\phi} $, which is a contradiction.
\end{itemize}
It follows that for all $ x \in K $,
$$
  \Int{G}{\Norm{\func{\phi}{y} - \alphArg{x}{\func{\phi}{x^{- 1} y}}}_{A}^{2}}{\mu}{y}
= \Int{K \Supp{\phi}}{\Norm{\func{\phi}{y} - \alphArg{x}{\func{\phi}{x^{- 1} y}}}_{A}^{2}}{\mu}{y}.
$$
Find $ G $-indexed families $ \Seq{U_{y}}{y \in G} $ and $ \Seq{V_{y}}{y \in G} $ of open subsets of $ G $ with the following properties:
\begin{itemize}
\item
$ U_{y} $ is an open neighborhood of $ e_{G} $ in $ G $ contained in $ K $ for each $ y \in G $.

\item
$ V_{y} $ is an open neighborhood of $ y $ in $ G $ for each $ y \in G $.

\item
For each $ y \in G $, we have
$$
\forall x \in U_{y}, ~ \forall z \in V_{y}: \qquad
\Norm{\func{\phi}{z} - \alphArg{x}{\func{\phi}{x^{- 1} z}}}_{A} < \frac{\epsilon}{\sqrt{1 + \func{\mu}{K \Supp{\phi}}}}.
$$
\end{itemize}
As $ K \Supp{\phi} $ is a compact subset of $ G $, and as $ \Set{V_{y}}{y \in G} $ covers $ K \Supp{\phi} $, there exists a finite subset $ F $ of $ G $ such that $ \Set{V_{y}}{y \in F} $ covers $ K \Supp{\phi} $. Let $ \ds x \in W \df \bigcap_{y' \in F} U_{y'} \subseteq K $, and let $ y \in K \Supp{\phi} $. Then $ \Pair{x}{y} \in U_{y'} \times V_{y'} $ for some $ y' \in F $, so
$$
\Norm{\func{\phi}{y} - \alphArg{x}{\func{\phi}{x^{- 1} y}}}_{A} < \frac{\epsilon}{\sqrt{1 + \func{\mu}{K \Supp{\phi}}}}.
$$
As $ y \in K \Supp{\phi} $ is arbitrary, this implies that
\begin{align*}
       \Int{G}{\Norm{\func{\phi}{y} - \alphArg{x}{\func{\phi}{x^{- 1} y}}}_{A}^{2}}{\mu}{y}
& =    \Int{K \Supp{\phi}}{\Norm{\func{\phi}{y} - \alphArg{x}{\func{\phi}{x^{- 1} y}}}_{A}^{2}}{\mu}{y} \\
& \leq \Int{K \Supp{\phi}}{\SqBr{\frac{\epsilon}{\sqrt{1 + \func{\mu}{K \Supp{\phi}}}}}^{2}}{\mu}{y} \\
& =    \func{\mu}{K \Supp{\phi}} \cdot \frac{\epsilon^{2}}{1 + \func{\mu}{K \Supp{\phi}}} \\
& <    \epsilon^{2}.
\end{align*}
Hence, for all $ x \in W $,
$$
  \Norm{\func{q}{\phi} - \FUNC{\func{\U}{x}}{\func{q}{\phi}}}_{\LL{G}{A}{\alpha}}
= \SqBr{\Int{G}{\Norm{\func{\phi}{y} - \alphArg{x}{\func{\phi}{x^{- 1} y}}}_{A}^{2}}{\mu}{y}}^{\frac{1}{2}}
< \Br{\epsilon^{2}}^{\frac{1}{2}}
= \epsilon.
$$
An $ \dfrac{\epsilon}{3} $-argument shows that for any $ \Phi \in \LL{G}{A}{\alpha} $ and any $ \epsilon > 0 $, there exists an open neighborhood $ O $ of $ e_{G} $ in $ G $ such that $ \Norm{\Phi - \FUNC{\func{\U}{x}}{\Phi}}_{\LL{G}{A}{\alpha}} < \epsilon $ for all $ x \in O $. Therefore, $ \U $ is strongly continuous.

\ul{The strong continuity of $ \V $}

Fix $ \phi \in \Cc{G,A} $ and $ \epsilon > 0 $. We will show that there exists an open neighborhood $ W $ of $ e_{\widehat{G}} $ in $ \widehat{G} $ such that
$$
\forall \varphi \in W: \qquad
\Norm{\func{q}{\phi} - \FUNC{\func{\V}{\varphi}}{\func{q}{\phi}}}_{\LL{G}{A}{\alpha}} < \epsilon.
$$
Observe for all $ \varphi \in \widehat{G} $ that
\begin{align*}
       \Norm{\func{q}{\phi} - \FUNC{\func{\V}{\varphi}}{\func{q}{\phi}}}_{\LL{G}{A}{\alpha}}
& =    \Norm{\func{q}{\Map{G}{A}{x}{\func{\phi}{x} - \func{\varphi}{x} \cdot \func{\phi}{x}}}}_{\LL{G}{A}{\alpha}} \\
& =    \Norm{\func{q}{\Map{G}{A}{x}{\SqBr{1 - \func{\varphi}{x}} \cdot \func{\phi}{x}}}}_{\LL{G}{A}{\alpha}} \\
& =    \Norm{
            \Int{G}
                {
                \alphArg{x^{- 1}}
                        {
                        \SqBr{\SqBr{1 - \func{\varphi}{x}} \cdot \func{\phi}{x}}^{\ast}
                        \SqBr{\SqBr{1 - \func{\varphi}{x}} \cdot \func{\phi}{x}}
                        }
                }
                {\mu}{x}
            }_{A}^{\frac{1}{2}} \\
& =    \Norm{
            \Int{G}{\alphArg{x^{- 1}}{\SqBr{1 - \func{\varphi}{x}}^{2} \cdot \func{\phi}{x}^{\ast} \func{\phi}{x}}}{\mu}{x}
            }_{A}^{\frac{1}{2}} \\
& \leq \SqBr{
            \Int{G}{\Norm{\alphArg{x^{- 1}}{\SqBr{1 - \func{\varphi}{x}}^{2} \cdot \func{\phi}{x}^{\ast} \func{\phi}{x}}}_{A}}{\mu}{x}
            }^{\frac{1}{2}} \\
& =    \SqBr{\Int{G}{\Norm{\SqBr{1 - \func{\varphi}{x}}^{2} \cdot \func{\phi}{x}^{\ast} \func{\phi}{x}}_{A}}{\mu}{x}}^{\frac{1}{2}} \\
& =    \SqBr{\Int{G}{\Abs{1 - \func{\varphi}{x}}^{2} \Norm{\func{\phi}{x}}_{A}^{2}}{\mu}{x}}^{\frac{1}{2}} \\
& =    \SqBr{\Int{\Supp{\phi}}{\Abs{1 - \func{\varphi}{x}}^{2} \Norm{\func{\phi}{x}}_{A}^{2}}{\mu}{x}}^{\frac{1}{2}}.
\end{align*}
The topology on $ \widehat{G} $ is the compact-open topology, i.e., is given by uniform convergence on compact subsets of $ G $, so we can pick an open neighborhood $ W $ of $ e_{\widehat{G}} $ such that for all $ \varphi \in W $ and $ x \in \Supp{\phi} $,
$$
\Abs{1 - \func{\varphi}{x}} < \dfrac{\epsilon}{\sqrt{\SqBr{1 + \func{\mu}{\Supp{\phi}}} \SqBr{1 + \Norm{\func{\phi}{x}}_{A}^{2}}}}.
$$
Then we have for all $ \varphi \in W $ that
\begin{align*}
       \SqBr{\Int{\Supp{\phi}}{\Abs{1 - \func{\varphi}{x}}^{2} \Norm{\func{\phi}{x}}_{A}^{2}}{\mu}{x}}^{\frac{1}{2}}
& \leq \SqBr{
            \Int{\Supp{\phi}}
                {
                \frac{\epsilon^{2}}{\SqBr{1 + \func{\mu}{\Supp{\phi}}} \SqBr{1 + \Norm{\func{\phi}{x}}_{A}^{2}}} \cdot
                \Norm{\func{\phi}{x}}_{A}^{2}
                }
                {\mu}{x}
            }^{\frac{1}{2}} \\
& \leq \SqBr{\Int{\Supp{\phi}}{\frac{\epsilon^{2}}{1 + \func{\mu}{\Supp{\phi}}}}{\mu}{x}}^{\frac{1}{2}} \\
& =    \SqBr{\func{\mu}{\Supp{\phi}} \cdot \frac{\epsilon^{2}}{1 + \func{\mu}{\Supp{\phi}}}}^{\frac{1}{2}} \\
& <    \Br{\epsilon^{2}}^{\frac{1}{2}} \\
& =    \epsilon.
\end{align*}
Hence, $ \Norm{\func{q}{\phi} - \FUNC{\func{\V}{\varphi}}{\func{q}{\phi}}}_{\LL{G}{A}{\alpha}} < \epsilon $ for all $ \varphi \in W $.

An $ \dfrac{\epsilon}{3} $-argument shows that for any $ \Phi \in \LL{G}{A}{\alpha} $ and any $ \epsilon > 0 $, there exists an open neighborhood $ O $ of $ e_{\widehat{G}} $ in $ \widehat{G} $ such that $ \Norm{\Phi - \FUNC{\func{\V}{\varphi}}{\Phi}}_{\LL{G}{A}{\alpha}} < \epsilon $ for all $ \varphi \in O $. Therefore, $ \V $ is strongly continuous.

\ul{$ \Pair{\U}{\V} $ satisfies the Weyl Commutation Relation for $ G $ on $ \LL{G}{A}{\alpha} $}

Observe for all $ x \in G $, $ \varphi \in \widehat{G} $, and $ \phi \in \Cc{G,A} $ that
\begin{align*}
    \FUNC{\func{\V}{\varphi} \func{\U}{x}}{\func{q}{\phi}}
& = \FUNC{\func{\V}{\varphi}}{\FUNC{\func{\U}{x}}{\func{q}{\phi}}} \\
& = \FUNC{\func{\V}{\varphi}}{\func{q}{\Map{G}{A}{y}{\alphArg{x}{\func{\phi}{x^{- 1} y}}}}} \\
& = \func{q}{\Map{G}{A}{y}{\func{\varphi}{y} \cdot \alphArg{x}{\func{\phi}{x^{- 1} y}}}} \\
& = \func{q}{\Map{G}{A}{y}{\func{\varphi}{x} \func{\varphi}{x^{- 1} y} \cdot \alphArg{x}{\func{\phi}{x^{- 1} y}}}} \\
& = \func{q}{\Map{G}{A}{y}{\func{\varphi}{x} \cdot \alphArg{x}{\func{\varphi}{x^{- 1} y} \cdot \func{\phi}{x^{- 1} y}}}} \\
& = \func{\varphi}{x} \cdot \func{q}{\Map{G}{A}{y}{\alphArg{x}{\func{\varphi}{x^{- 1} y} \cdot \func{\phi}{x^{- 1} y}}}} \\
& = \func{\varphi}{x} \cdot \func{q}{\Map{G}{A}{y}{\alphArg{x}{\Func{\varphi \cdot \phi}{x^{- 1} y}}}} \\
& = \func{\varphi}{x} \cdot \FUNC{\func{\U}{x}}{\func{q}{\varphi \cdot \phi}} \\
& = \func{\varphi}{x} \cdot \FUNC{\func{\U}{x}}{\FUNC{\func{\V}{\varphi}}{\func{q}{\phi}}} \\
& = \func{\varphi}{x} \cdot \FUNC{\func{\U}{x} \func{\V}{\varphi}}{\func{q}{\phi}},
\end{align*}
so by continuity, $ \func{\V}{\varphi} \func{\U}{x} = \func{\varphi}{x} \cdot \func{\U}{x} \func{\V}{\varphi} $ for all $ x \in G $ and $ \varphi \in \widehat{G} $.

\ul{$ \Trip{\LL{G}{A}{\alpha}}{\M}{\U} $ is a $ \Trip{G}{A}{\alpha} $-covariant modular representation}

Observe for all $ x \in G $, $ a \in A $, and $ \phi \in \Cc{G,A} $ that
\begin{align*}
    \FUNC{\func{\U}{x} \func{\M}{a}}{\func{q}{\phi}}
& = \FUNC{\func{\U}{x}}{\FUNC{\func{\M}{a}}{\func{q}{\phi}}} \\
& = \FUNC{\func{\U}{x}}{\func{q}{\Map{G}{A}{y}{a \func{\phi}{y}}}} \\
& = \func{q}{\Map{G}{A}{y}{\alphArg{x}{a \func{\phi}{x^{- 1} y}}}} \\
& = \func{q}{\Map{G}{A}{y}{\alphArg{x}{a} \alphArg{x}{\func{\phi}{x^{- 1} y}}}} \\
& = \FUNC{\func{\M}{\alphArg{x}{a}}}{\func{q}{\Map{G}{A}{y}{\alphArg{x}{\func{\phi}{x^{- 1} y}}}}} \\
& = \FUNC{\func{\M}{\alphArg{x}{a}}}{\FUNC{\func{\U}{x}}{\func{q}{\phi}}} \\
& = \FUNC{\func{\M}{\alphArg{x}{a}} \func{\U}{x}}{\func{q}{\phi}},
\end{align*}
so by continuity, $ \func{\U}{x} \func{\M}{a} = \func{\M}{\alphArg{x}{a}} \func{\U}{x} $ for all $ x \in G $ and $ a \in A $.

\ul{$ \Trip{\LL{G}{A}{\alpha}}{\M}{\V} $ is a $ \Trip{\widehat{G}}{A}{\iota} $-covariant modular representation}

Observe for all $ \varphi \in \widehat{G} $, $ a \in A $, and $ \phi \in \Cc{G,A} $ that
\begin{align*}
    \FUNC{\func{\V}{\varphi} \func{\M}{a}}{\func{q}{\phi}}
& = \FUNC{\func{\V}{\varphi}}{\FUNC{\func{\M}{a}}{\func{q}{\phi}}} \\
& = \FUNC{\func{\V}{\varphi}}{\func{q}{\Map{G}{A}{y}{a \func{\phi}{y}}}} \\
& = \func{q}{\Map{G}{A}{y}{\func{\varphi}{y} \cdot a \func{\phi}{y}}} \\
& = \func{q}{\Map{G}{A}{y}{a \SqBr{\Func{\varphi \cdot \phi}{y}}}} \\
& = \FUNC{\func{\M}{a}}{\func{q}{\varphi \cdot \phi}} \\
& = \FUNC{\func{\M}{a}}{\FUNC{\func{\V}{\varphi}}{\func{q}{\phi}}} \\
& = \FUNC{\func{\M}{a} \func{\V}{\varphi}}{\func{q}{\phi}},
\end{align*}
so by continuity, $ \func{\V}{\varphi} \func{\M}{a} = \func{\M}{a} \func{\V}{\varphi} $ for all $ \varphi \in \widehat{G} $ and $ a \in A $.
\end{proof}


The ultimate goal of this section is to establish the following proposition, which we presently state in an imprecise form.



\begin{Prop} \label{An Injective Map from the Class of Heisenberg Modular Representations to the Class of Covariant Modular Representations}
There is an injective map from the class of $ \Trip{G}{A}{\alpha} $-Heisenberg modular representations to the class of $ \Trip{G}{\Co{G,A}}{\lt \otimes \alpha} $-covariant modular representations, where $ \lt $ denotes the action of $ G $ on $ \Co{G,A} $ by left translations, and $ \lt \otimes \alpha $ denotes the action of $ G $ on $ \Co{G,A} $ defined by
$$
\forall x \in G, ~ \forall g \in \Co{G,A}: \qquad
\func{\Br{\lt \otimes \alpha}_{x}}{g} \df \Map{G}{A}{y}{\alphArg{x}{\func{g}{x^{- 1} y}}}.
$$
\end{Prop}


The proposition is imprecisely stated because we have not yet specified what the injective map is, but this will be explicated in due course.

The main tool for proving the proposition is a $ C^{\ast} $-algebra-valued version of the Fourier transform, which we will introduce soon. In order to show that this generalized Fourier transform is well-defined, the following approximation lemma is indispensable.



\begin{Lem} \label{The Approximation Lemma}
Let $ X $ be a locally compact Hausdorff space, $ V $ a normed vector space, and $ D $ a dense subset of $ V $. Then for any $ f \in \Cc{X,V} $ and $ \epsilon > 0 $, there exist $ \varphi_{1},\ldots,\varphi_{n} \in \Cc{X} $ and $ v_{1},\ldots,v_{n} \in D $ such that
$$
\forall x \in X: \qquad
\Norm{\func{f}{x} - \sum_{i = 1}^{n} \func{\varphi_{i}}{x} \cdot v_{i}}_{V} < \epsilon.
$$
If $ \nu $ is a regular Borel measure on $ X $, then for any $ f \in \Cc{X,V} $ and $ \epsilon > 0 $, there exist $ \varphi_{1},\ldots,\varphi_{n} \in \Cc{X} $ and $ v_{1},\ldots,v_{n} \in D $ such that
$$
\Int{X}{\Norm{\func{f}{x} - \sum_{i = 1}^{n} \func{\varphi_{i}}{x} \cdot v_{i}}_{V}}{\nu}{x} < \epsilon.
$$
\end{Lem}


This is a folklore result that can be straightforwardly proven using partitions of unity. To avoid disrupting the flow of this paper, we will provide a proof of it in the appendix.



\begin{Def} \label{The Generalized Fourier Transform}
The $ A $-valued \emph{generalized Fourier transform} for $ G $ with respect to a Haar measure $ \nu $ on $ \widehat{G} $ is the map $ \F{\nu}{A}: \Cc{\widehat{G},A} \to \Co{G,A} $ defined by
$$
\forall f \in \Cc{\widehat{G},A}: \qquad
\FT{\nu}{A}{f} \df \Map{G}{A}{x}{\Int{\widehat{G}}{\func{\widehat{x}}{\varphi} \cdot \func{f}{\varphi}}{\nu}{\varphi}}.
$$
\end{Def}


We proceed to demonstrate the consistency of this definition.

When $ A \neq \C $, it is not at all obvious why the image of $ \F{\nu}{A} $ should be in $ \Co{G,A} $. To see this, let us pick $ f \in \Cc{\widehat{G},A} $. For every $ x \in G $, the integrand of $ \ds \Int{\widehat{G}}{\func{\widehat{x}}{\varphi} \cdot \func{f}{\varphi}}{\nu}{\varphi} $ belongs to $ \Cc{\widehat{G},A} $, so the integral exists. Furthermore, for all $ x \in G $,
\begin{align*}
       \Norm{\Int{\widehat{G}}{\func{\widehat{x}}{\varphi} \cdot \func{f}{\varphi}}{\nu}{\varphi}}_{A}
& \leq \Int{\widehat{G}}{\Norm{\func{\widehat{x}}{\varphi} \cdot \func{f}{\varphi}}_{A}}{\nu}{\varphi} \\
& =    \Int{\widehat{G}}{\Abs{\func{\widehat{x}}{\varphi}} \Norm{\func{f}{\varphi}}_{A}}{\nu}{\varphi} \\
& =    \Int{\widehat{G}}{\Norm{\func{f}{\varphi}}_{A}}{\nu}{\varphi} \\
& =    \Norm{f}_{\nu,1},
\end{align*}
so $ \FT{\nu}{A}{f} $ is a function from $ G $ to $ A $ that is pointwise-bounded by $ \Norm{f}_{\nu,1} $. To check that it is also continuous, fix $ x \in G $ and $ \epsilon > 0 $. As $ \Supp{f} $ is a compact subset of $ \widehat{G} $ with respect to the compact-open topology on $ \CC{G} $, the Arzel\`a-Ascoli Theorem says that $ \Supp{f} $ is an equicontinuous subset of $ \CC{G} $, so there exists an open neighborhood $ U $ of $ x $ in $ G $ such that for all $ y \in U $ and $ \varphi \in \Supp{f} $,
$$
\Abs{\func{\varphi}{x} - \func{\varphi}{y}} < \dfrac{\epsilon}{1 + \Norm{f}_{\nu,1}}.
$$
Consequently, for all $ y \in U $,
\begin{align*}
       \Norm{\FUNC{\FT{\nu}{A}{f}}{x} - \FUNC{\FT{\nu}{A}{f}}{y}}_{A}
& =    \Norm{
            \Int{\widehat{G}}{\func{\widehat{x}}{\varphi} \cdot \func{f}{\varphi}}{\nu}{\varphi} -
            \Int{\widehat{G}}{\func{\widehat{y}}{\varphi} \cdot \func{f}{\varphi}}{\nu}{\varphi}
            }_{A} \\
& =    \Norm{\Int{\widehat{G}}{\SqBr{\func{\widehat{x}}{\varphi} - \func{\widehat{y}}{\varphi}} \cdot \func{f}{\varphi}}{\nu}{\varphi}}_{A} \\
& =    \Norm{\Int{\widehat{G}}{\SqBr{\func{\varphi}{x} - \func{\varphi}{y}} \cdot \func{f}{\varphi}}{\nu}{\varphi}}_{A} \\
& \leq \Int{\widehat{G}}{\Norm{\SqBr{\func{\varphi}{x} - \func{\varphi}{y}} \cdot \func{f}{\varphi}}_{A}}{\nu}{\varphi} \\
& =    \Int{\widehat{G}}{\Abs{\func{\varphi}{x} - \func{\varphi}{y}} \Norm{\func{f}{\varphi}}_{A}}{\nu}{\varphi} \\
& =    \Int{\Supp{f}}{\Abs{\func{\varphi}{x} - \func{\varphi}{y}} \Norm{\func{f}{\varphi}}_{A}}{\nu}{\varphi} \\
& \leq \Int{\Supp{f}}{\frac{\epsilon}{1 + \Norm{f}_{\nu,1}} \cdot \Norm{\func{f}{\varphi}}_{A}}{\nu}{\varphi} \\
& =    \frac{\epsilon}{1 + \Norm{f}_{\nu,1}} \Int{\Supp{f}}{\Norm{\func{f}{\varphi}}_{A}}{\nu}{\varphi} \\
& =    \frac{\epsilon}{1 + \Norm{f}_{\nu,1}} \cdot \Norm{f}_{\nu,1} \\
& <    \epsilon.
\end{align*}
As $ x \in G $ is arbitrary, this proves that $ \FT{\nu}{A}{f} $ is continuous, so the image of $ \F{\nu}{A} $ is contained in $ \Cb{G,A} $.

Now, given an $ f \in \Cc{\widehat{G}} $ and an $ a \in A $, we have for all $ x \in G $ that
\begin{align*}
    \FUNC{\FT{\nu}{A}{f \diamond a}}{x}
& = \Int{\widehat{G}}{\func{\widehat{x}}{\varphi} \cdot \Func{f \diamond a}{\varphi}}{\nu}{\varphi} \\
& = \Int{\widehat{G}}{\func{\widehat{x}}{\varphi} \cdot \SqBr{\func{f}{\varphi} \cdot a}}{\nu}{\varphi} \\
& = \SqBr{\Int{\widehat{G}}{\func{\widehat{x}}{\varphi} \func{f}{\varphi}}{\nu}{\varphi}} \cdot a \\
& = \FUNC{\FT{\nu}{\C}{f}}{x} \cdot a.
\end{align*}
As we already know that $ \FT{\nu}{\C}{f} \in \Co{G} $, we get $ \FT{\nu}{A}{f \diamond a} = \FT{\nu}{\C}{f} \diamond a \in \Co{G,A} $. Hence, as $ f \in \Cc{\widehat{G}} $ and $ a \in A $ are arbitrary, we obtain
$$
          \Im{\F{\nu}{A}}{\Span{\Cc{\widehat{G}} \diamond A}}
\subseteq \Span{\Co{G} \diamond A}
\subseteq \Co{G,A}.
$$
Let $ f \in \Cc{\widehat{G},A} $. \autoref{The Approximation Lemma} makes it possible to find a sequence $ \Seq{f_{n}}{n \in \N} $ in $ \Span{\Cc{\widehat{G}} \diamond A} $ such that $ \ds \lim_{n \to \infty} \Norm{f - f_{n}}_{\nu,1} = 0 $. Then because $ \Norm{\FT{\nu}{A}{f - f_{n}}}_{\Cb{G,A}} \leq \Norm{f - f_{n}}_{\nu,1} $ for all $ n \in \N $, we get
$$
  \lim_{n \to \infty} \Norm{\FT{\nu}{A}{f} - \FT{\nu}{A}{f_{n}}}_{\Cb{G,A}}
= \lim_{n \to \infty} \Norm{\FT{\nu}{A}{f - f_{n}}}_{\Cb{G,A}}
= 0.
$$
However, as seen above, $ \FT{\nu}{A}{f_{n}} \in \Co{G,A} $ for all $ n \in \N $, so because $ \Co{G,A} $ is complete with respect to the supremum norm, it follows that $ \FT{\nu}{A}{f} \in \Co{G,A} $. As $ f \in \Cc{\widehat{G},A} $ is arbitrary, we have proven that $ \F{\nu}{A} $ maps $ \Cc{\widehat{G},A} $ to $ \Co{G,A} $.

Knowing now that $ \F{\nu}{A}: \Cc{\widehat{G},A} \to \Co{G,A} $ is well-defined, our next step is to show the following.



\begin{Prop} \label{The Generalized Fourier Transform Extends to a C*-Isomorphism}
$ \F{\nu}{A} $ extends to a $ C^{\ast} $-isomorphism $ \overline{\F{\nu}{A}}: \Cstar{\widehat{G},A,\iota} \to \Co{G,A} $.
\end{Prop}

\begin{proof}
It is routine to check that $ \F{\nu}{A}: \Trip{\Cc{\widehat{G},A}}{\star_{\widehat{G},A,\iota}}{^{\ast_{\star_{\widehat{G},A,\iota}}}} \to \Co{G,A} $ is a $ \ast $-homomorphism. As we know that $ \F{\nu}{A} $ is contractive with respect to $ \Norm{\cdot}_{\nu,1} $, the theory of $ C^{\ast} $-crossed products says that $ \F{\nu}{A} $ extends to a $ C^{\ast} $-homomorphism $ \overline{\F{\nu}{A}}: \Cstar{\widehat{G},A,\iota} \to \Co{G,A} $.

By Lemma 2.73 of \cite{Wi} and the theory of $ C^{\ast} $-tensor products, we have the series of $ C^{\ast} $-isomorphisms
$$
      \Cstar{\widehat{G},A,\iota}
\cong \Cstar{\widehat{G}} \otimes A
\cong \Co{G} \otimes A
\cong \Co{G,A},
$$
which are implemented as follows: For all $ f_{1},\ldots,f_{n} \in \Cc{\widehat{G}} $ and $ a_{1},\ldots,a_{n} \in A $,
$$
        \func{\eta_{\Trip{\widehat{G}}{A}{\iota}}}{\sum_{i = 1}^{n} f_{i} \diamond a_{i}}
\mapsto \sum_{i = 1}^{n} \func{\eta_{\widehat{G}}}{f_{i}} \odot a_{i}
\mapsto \sum_{i = 1}^{n} \FT{\nu}{\C}{f_{i}} \odot a_{i}
\mapsto \sum_{i = 1}^{n} \FT{\nu}{\C}{f_{i}} \diamond a_{i}.
$$
However, $ \ds \FT{\nu}{A}{\sum_{i = 1}^{n} f_{i} \diamond a_{i}} = \sum_{i = 1}^{n} \FT{\nu}{\C}{f_{i}} \diamond a_{i} $, so $ \overline{\F{\nu}{A}} $ agrees with some $ C^{\ast} $-isomorphism from $ \Cstar{\widehat{G},A,\iota} $ to $ \Co{G,A} $ on a dense subset. It is therefore precisely that $ C^{\ast} $-isomorphism.
\end{proof}



\begin{Def} \label{A Non-Degenerate *-Representation of C0(G,A) on X}
For a $ \Trip{\widehat{G}}{A}{\iota} $-covariant modular representation $ \Trip{\X}{\rho}{S} $, let
$$
\pi^{\X,\rho,S}_{\nu} \df \overline{\Pi}_{\X,\rho,S} \circ \overline{\F{\nu}{A}}^{- 1}: \Co{G,A} \to \Adj{\X},
$$
which is a non-degenerate $ \ast $-representation of $ \Co{G,A} $ on $ \X $.
\end{Def}


Finally, let us tackle the main objective of this section.


\begin{proof}[Proof of \autoref{An Injective Map from the Class of Heisenberg Modular Representations to the Class of Covariant Modular Representations}]
Fixing a Haar measure $ \nu $ on $ \widehat{G} $, we divide the proof into two parts.

\ul{Defining the desired class map}

Observe for all $ x,y \in G $ and $ f \in \Cc{\widehat{G},A} $ that
\begin{align*}
    \FUNC{\func{\Br{\lt \otimes \alpha}_{x}}{\FT{\nu}{A}{f}}}{y}
& = \alphArg{x}{\FUNC{\FT{\nu}{A}{f}}{x^{- 1} y}} \\
& = \alphArg{x}{\Int{\widehat{G}}{\func{\widehat{\Br{x^{- 1} y}}}{\varphi} \cdot \func{f}{\varphi}}{\nu}{\varphi}} \\
& = \Int{\widehat{G}}{\func{\widehat{\Br{x^{- 1} y}}}{\varphi} \cdot \alphArg{x}{\func{f}{\varphi}}}{\nu}{\varphi} \\
& = \Int{\widehat{G}}{\func{\widehat{y}}{\varphi} \func{\widehat{x^{- 1}}}{\varphi} \cdot \alphArg{x}{\func{f}{\varphi}}}{\nu}{\varphi} \\
& = \Int{\widehat{G}}{\func{\widehat{y}}{\varphi} \cdot \FUNC{\widehat{x^{- 1}} \cdot \Br{\alph{x} \circ f}}{\varphi}}{\nu}{\varphi} \\
& = \FUNC{\FT{\nu}{A}{\widehat{x^{- 1}} \cdot \Br{\alph{x} \circ f}}}{y}.
\end{align*}
Hence, $ \func{\Br{\lt \otimes \alpha}_{x}}{\FT{\nu}{A}{f}} = \FT{\nu}{A}{\widehat{x^{- 1}} \cdot \Br{\alph{x} \circ f}} $ for all $ x \in G $ and $ f \in \Cc{\widehat{G},A} $.

Given a $ \Trip{G}{A}{\alpha} $-Heisenberg modular representation $ \Quad{\X}{\rho}{R}{S} $, we will exploit the computation above to show that $ \Trip{\X}{\pi^{\X,\rho,S}_{\nu}}{R} $ is a $ \Trip{G}{\Co{G,A}}{\lt \otimes \alpha} $-covariant modular representation.

Firstly, $ \rho $ is a non-degenerate $ \ast $-representation of $ A $ on $ \X $, so $ \Pi_{\X,\rho,S} $ is a non-degenerate $ \ast $-representation of $ \Cstar{\widehat{G},A,\iota} $ on $ \X $, which, in turn, means that $ \pi^{\X,\rho,S}_{\nu} $ is a non-degenerate $ \ast $-representation of $ \Co{G,A} $ on $ \X $. Secondly, we have for all $ x \in G $ and $ f \in \Cc{\widehat{G},A} $ that
\begin{align*}
    \func{R}{x} \func{\pi^{\X,\rho,S}_{\nu}}{\FT{\nu}{A}{f}}
& = \func{R}{x} \SqBr{\Int{\widehat{G}}{\func{\rho}{\func{f}{\varphi}} \func{S}{\varphi}}{\nu}{\varphi}} \\
& = \Int{\widehat{G}}{\func{R}{x} \func{\rho}{\func{f}{\varphi}} \func{S}{\varphi}}{\nu}{\varphi} \\
& = \Int{\widehat{G}}{\func{\rho}{\alphArg{x}{\func{f}{\varphi}}} \func{R}{x} \func{S}{\varphi}}{\nu}{\varphi} \\
& = \Int{\widehat{G}}{\func{\varphi}{x^{- 1}} \cdot \func{\rho}{\alphArg{x}{\func{f}{\varphi}}} \func{S}{\varphi} \func{R}{x}}{\nu}{\varphi} \\
& = \SqBr{\Int{\widehat{G}}{\func{\varphi}{x^{- 1}} \cdot \func{\rho}{\alphArg{x}{\func{f}{\varphi}}} \func{S}{\varphi}}{\nu}{\varphi}}
    \func{R}{x} \\
& = \SqBr{\Int{\widehat{G}}{\func{\rho}{\func{\varphi}{x^{- 1}} \cdot \alphArg{x}{\func{f}{\varphi}}} \func{S}{\varphi}}{\nu}{\varphi}}
    \func{R}{x} \\
& = \SqBr{
         \Int{\widehat{G}}{\func{\rho}{\func{\widehat{x^{- 1}}}{\varphi} \cdot \Func{\alph{x} \circ f}{\varphi}} \func{S}{\varphi}}
             {\nu}{\gamma}
         }
    \func{R}{x} \\
& = \SqBr{\Int{\widehat{G}}{\func{\rho}{\FUNC{\widehat{x^{- 1}} \cdot \Br{\alph{x} \circ f}}{\varphi}} \func{S}{\varphi}}{\nu}{\varphi}}
    \func{R}{x} \\
& = \func{\pi^{\X,\rho,S}_{\nu}}{\FT{\nu}{A}{\widehat{x^{- 1}} \cdot \Br{\alph{x} \circ f}}} \func{R}{x} \\
& = \func{\pi^{\X,\rho,S}_{\nu}}{\func{\Br{\lt \otimes \alpha}_{x}}{\FT{\nu}{A}{f}}} \func{R}{x}.
\end{align*}
As the image of $ \F{\nu}{A} $ is dense in $ \Co{G,A} $, it follows from continuity that for all $ x \in G $ and $ g \in \Co{G,A} $,
$$
\func{R}{x} \func{\pi^{\X,\rho,S}_{\nu}}{g} = \func{\pi^{\X,\rho,S}_{\nu}}{\func{\Br{\lt \otimes \alpha}_{x}}{g}} \func{R}{x}.
$$
Hence, $ \Trip{\X}{\pi^{\X,\rho,S}_{\nu}}{R} $ is a $ \Trip{G}{\Co{G,A}}{\lt \otimes \alpha} $-covariant modular representation. We can thus define a map from the class of $ \Trip{G}{A}{\alpha} $-Heisenberg modular representations to the class of $ \Trip{G}{\Co{G,A}}{\lt \otimes \alpha} $-covariant modular representation according to the rule
$$
\Quad{\X}{\rho}{R}{S} \mapsto \Trip{\X}{\pi^{\X,\rho,S}_{\nu}}{R}.
$$

\ul{Injectivity of the class map}

Let $ \Quad{\X_{1}}{\rho_{1}}{R_{1}}{S_{1}} $ and $ \Quad{\X_{2}}{\rho_{2}}{R_{2}}{S_{2}} $ be $ \Trip{G}{A}{\alpha} $-Heisenberg modular representations such that
$$
\Trip{\X_{1}}{\pi^{\X_{1},\rho_{1},S_{1}}_{\nu}}{R_{1}} = \Trip{\X_{2}}{\pi^{\X_{2},\rho_{2},S_{2}}_{\nu}}{R_{2}}.
$$
Clearly, $ \X_{1} = \X_{2} $, $ R_{1} = R_{2} $, and $ \pi^{\X_{1},\rho_{1},S_{1}}_{\nu} = \pi^{\X_{2},\rho_{2},S_{2}}_{\nu} $. Hence,
$$
  \overline{\Pi}_{\X_{1},\rho_{1},S_{1}}
= \pi^{\X_{1},\rho_{1},S_{1}}_{\nu} \circ \overline{\F{\nu}{A}}
= \pi^{\X_{2},\rho_{2},S_{2}}_{\nu} \circ \overline{\F{\nu}{A}}
= \overline{\Pi}_{\X_{2},\rho_{2},S_{2}},
$$
which yields $ \Pi_{\X_{1},\rho_{1},S_{1}} = \Pi_{\X_{2},\rho_{2},S_{2}} $. By \autoref{Recovering a Covariant Modular Representation from Its Integrated Form}, $ \rho_{1} = \rho_{2} $ and $ S_{1} = S_{2} $. Therefore, the proposed class map is indeed injective.
\end{proof}


Actually, one can show that the image of the class map above is the class of $ \Trip{G}{\Co{G,A}}{\lt \otimes \alpha} $-covariant modular representations whose underlying Hilbert $ C^{\ast} $-module is a full Hilbert $ A $-module. However, we will have no need of this fact.



\section{Hilbert $ \Comp{\H} $-Modules}


In this section, we take a brief excursion into Hilbert $ \Comp{\H} $-modules. The initial material can be found in \cite{BaGu1,BaGu2}, but we have decided to supply our own proofs, some of which are simpler than the original ones.

Throughout this section, we shall fix a non-trivial Hilbert space $ \H $.



\begin{Lem} \label{The Linear Functional Associated To a Rank-One Projection}
Let $ P $ be a rank-one projection on $ \H $. Then there exists a unique linear functional $ f $ on $ \Comp{\H} $ such that $ P S P = \func{f}{S} \cdot P $ for all $ S \in \Comp{\H} $.
\end{Lem}

\begin{proof}
Let $ \K $ denote the range of $ P $, which is a one-dimensional subspace of $ \H $. Given $ S \in \Comp{\H} $, the map $ \Br{P S P}|_{\K} $ is a linear operator on $ \K $, so it is a unique scalar multiple of $ \Id_{\K} = P|_{\K} $. Hence, there is a unique function $ f: \Comp{\H} \to \C $ such that $ \Br{P S P}|_{\K} = \func{f}{S} \cdot P|_{\K} $ for all $ S \in \Comp{\H} $. Then
$$
  \Func{P S P}{v}
= \Func{P S P}{\func{P}{v}}
= \func{f}{S} \cdot \func{P}{\func{P}{v}}
= \func{f}{S} \cdot \func{P}{v}
$$
for all $ S \in \Comp{\H} $ and $ v \in \H $, which means that $ P S P = \func{f}{S} \cdot P $. Finally, observe for all $ S,T \in \Comp{\H} $ and $ \lambda \in \C $ that
\begin{align*}
    \func{f}{S + \lambda \cdot T} \cdot P
& = P \Br{S + \lambda \cdot T} P \\
& = P S P + \lambda \cdot \Br{P T P} \\
& = \func{f}{S} \cdot P + \lambda \func{f}{T} \cdot P \\
& = \SqBr{\func{f}{S} + \lambda \func{f}{T}} \cdot P.
\end{align*}
As $ \func{f}{S + \lambda \cdot T} $ is the unique $ \kappa \in \C $ for which $ P \Br{S + \lambda \cdot T} P = \kappa \cdot P $, we get $ \func{f}{S + \lambda \cdot T} = \func{f}{S} + \lambda \func{f}{T} $. Therefore, $ f $ is the unique linear functional on $ \Comp{\H} $ such that $ P S P = \func{f}{S} \cdot P $ for all $ S \in \Comp{\H} $.
\end{proof}


Let $ \mathcal{P}_{1} $ denote the set of all rank-one projections on $ \H $. \autoref{The Linear Functional Associated To a Rank-One Projection} says that there exists a $ \mathcal{P}_{1} $-indexed family $ \Seq{f_{P}}{P \in \mathcal{P}_{1}} $ of linear functionals on $ \H $ such that
$$
\forall P \in \mathcal{P}_{1}, ~ \forall S \in \Comp{\H}: \qquad
P S P = \func{f_{P}}{S} \cdot P.
$$
These linear functionals play a pivotal role in the next result.



\begin{Thm} \label{The Hilbert Space Coming from a Rank-One Projection}
Let $ \X $ be a non-trivial Hilbert $ \Comp{\H} $-module, and $ P $ a rank-one projection on $ \H $. Then $ \X \bullet P $ is a non-trivial closed subspace of $ \X \bullet P $ that has the structure of a Hilbert space, whose inner product $ \Inner{\cdot}{\cdot}_{\X \bullet P} $ is given by
$$
\forall \zeta,\eta \in \X \bullet P: \qquad
\Inner{\zeta}{\eta}_{\X \bullet P} \df \func{f_{P}}{\Inner{\zeta}{\eta}_{\X}}.
$$
Furthermore, the norm on $ \X \bullet P $ induced by $ \Inner{\cdot}{\cdot}_{\X \bullet P} $ coincides with the restriction of $ \Norm{\cdot}_{\X} $ to $ \X \bullet P $.
\end{Thm}

\begin{proof}
It is clear that $ \X \bullet P $ is a subspace of $ \X $. As $ \X $ is non-trivial, $ \Cl{\Span{\Inner{\X}{\X}_{\X}}}{\Comp{\H}} $ is a non-trivial ideal of $ \Comp{\H} $, but as $ \Comp{\H} $ is a simple $ C^{\ast} $-algebra, we have $ \Cl{\Span{\Inner{\X}{\X}_{\X}}}{\Comp{\H}} = \Comp{\H} $. Hence,
\begin{align*}
              P
& \in         P \Comp{\H} P \\
& =           P \Cl{\Span{\Inner{\X}{\X}_{\X}}}{\Comp{\H}} P \\
& \subseteq   \Cl{P \Span{\Inner{\X}{\X}_{\X}} P}{\Comp{\H}} \\
& =           \Cl{\Span{\Inner{\X \bullet P}{\X \bullet P}_{\X}}}{\Comp{\H}},
\end{align*}
which implies that $ \X \bullet P $ is a non-trivial subspace of $ \X $.

To see that $ \X \bullet P $ is a closed subspace of $ \X $, suppose that $ \Seq{\zeta_{n}}{n \in \N} $ is a sequence in $ \X \bullet P $ that converges to some $ \eta \in \X $. Then because $ \zeta_{n} \bullet P = \zeta_{n} $ for all $ n \in \N $, we have
$$
  \eta
= \lim_{n \to \infty} \zeta_{n}
= \lim_{n \to \infty} \zeta_{n} \bullet P
= \Br{\lim_{n \to \infty} \zeta_{n}} \bullet P
= \eta \bullet P.
$$
Hence, $ \eta \in \X \bullet P $, which proves that $ \X \bullet P $ is a closed subspace of $ \X $.

Clearly, $ \Inner{\cdot}{\cdot}_{\X \bullet P} $ is a sesquilinear form on $ \X \bullet P $, so it remains to see that it is positive definite and complete. Let $ \zeta \in \X \bullet P $. Then $ \Inner{\zeta}{\zeta}_{\X} $ is positive in $ \Comp{\H} $, which means that
$$
  \func{f_{P}}{\Inner{\zeta}{\zeta}_{\X}} \cdot P
= P \Inner{\zeta}{\zeta}_{\X} P
= P \Inner{\zeta}{\zeta}_{\X} P^{\ast}
$$
is positive in $ \Comp{\H} $ as well. As $ \Id_{\H} - P $ is not invertible in $ \Comp{\H}^{\sim} $, we have $ 1 \in \func{\sigma_{\Comp{\H}}}{P} $. Hence,
$$
          \func{f_{P}}{\Inner{\zeta}{\zeta}_{\X}}
\in       \func{\sigma_{\Comp{\H}}}{\func{f_{P}}{\Inner{\zeta}{\zeta}_{\X}} \cdot P}
\subseteq \R_{\geq 0},
$$
which demonstrates that $ \Inner{\cdot}{\cdot}_{\X \bullet P} $ is at least positive semidefinite. Next, observe for all $ \zeta,\eta \in \X \bullet P $ that
\begin{align*}
    \Abs{\Inner{\zeta}{\eta}_{\X \bullet P}}
& = \Abs{\func{f_{P}}{\Inner{\zeta}{\eta}_{\X}}} \\
& = \Norm{\func{f_{P}}{\Inner{\zeta}{\eta}_{\X}} \cdot P}_{\Comp{\H}} \qquad \Br{\text{As $ \Norm{P}_{\Comp{\H}} = 1 $.}}  \\
& = \Norm{P \Inner{\zeta}{\eta}_{\X} P}_{\Comp{\H}} \\
& = \Norm{\Inner{\zeta \bullet P}{\eta \bullet P}_{\X}}_{\Comp{\H}} \\
& = \Norm{\Inner{\zeta}{\eta}_{\X}}_{\Comp{\H}}. \qquad \Br{\text{As $ \zeta \bullet P = \zeta $ and $ \eta \bullet P = \eta $.}}
\end{align*}
Consequently, if $ \Inner{\zeta}{\zeta}_{\X \bullet P} = 0 $ for some $ \zeta \in \X \bullet P $, then $ \Inner{\zeta}{\zeta}_{\X} = 0_{\Comp{\H}} $, which yields $ \zeta = 0_{\X} = 0_{\X \bullet P} $. This proves that $ \Inner{\cdot}{\cdot}_{\X \bullet P} $ is positive definite. Incidentally, this also proves that $ \Norm{\zeta}_{\X \bullet P} = \Norm{\zeta}_{\X} $ for all $ \zeta \in \X \bullet P $. As $ \X \bullet P $ is a closed subspace of $ \X $, it is a Banach space with respect to the restriction of $ \Norm{\cdot}_{\X} $ to $ \X \bullet P $, and is thus a Banach space with respect to $ \Norm{\cdot}_{\X \bullet P} $. Therefore, $ \X \bullet P $ is a Hilbert space whose inner product is given by $ \Inner{\cdot}{\cdot}_{\X \bullet P} $, and the Hilbert-space norm on $ \X \bullet P $ is precisely the restriction of $ \Norm{\cdot}_{\X} $ to $ \X \bullet P $.
\end{proof}



\begin{Thm} \label{The Reduction of a Submodule of a Hilbert K(H)-Module by a Rank-One Projection Generates the Submodule}
Let $ \X $ be a Hilbert $ \Comp{\H} $-module, $ \Y $ a $ \Comp{\H} $-submodule of $ \X $ that is not necessarily closed, and $ P $ a rank-one projection on $ \H $. Then the closed $ \Comp{\H} $-linear span of $ \Y \bullet P $ in $ \X $ is the closure $ \Cl{\Y}{\X} $ of $ \Y $ in $ \X $.
\end{Thm}

\begin{proof}
As $ \Y $ is a $ \Comp{\H} $-submodule of $ \X $, we can see that $ \Y \bullet P \subseteq \Y $ and that the $ \Comp{\H} $-linear span of $ \Y \bullet P $ is contained in $ \Y $. Hence, the closed $ \Comp{\H} $-linear span of $ \Y \bullet P $ in $ \X $ is contained in $ \Cl{\Y}{\X} $. It thus remains to establish the reverse inclusion.

Let $ \zeta \in \Y $, and let $ F $ be a rank-$ n $ operator on $ \H $. We claim that $ \zeta \bullet F $ belongs to the $ \Comp{\H} $-linear span of $ \Y \bullet P $. Let $ \Seq{v_{i}}{i \in \SqBr{n}} $ be an orthonormal basis of $ \Range{F} $, and $ v $ a unit vector in $ \Range{P} $. Then
\begin{align*}
    F
& = \Proj{\H}{\Range{P}} F \\
& = \Br{\sum_{i = 1}^{n} \Ket{v_{i}} \Bra{v_{i}}} F \\
& = \Br{\sum_{i = 1}^{n} \Ket{v_{i}} \Bra{v} \Ket{v} \Bra{v} \Ket{v} \Bra{v_{i}}} F \\
& = \Br{\sum_{i = 1}^{n} \Ket{v_{i}} \Bra{v} P \Ket{v} \Bra{v_{i}}} F \\
& = \sum_{i = 1}^{n} \Ket{v_{i}} \Bra{v} P \Ket{v} \Bra{v_{i}} F.
\end{align*}
Hence,
\begin{align*}
    \zeta \bullet F
& = \zeta \bullet \Br{\sum_{i = 1}^{n} \Ket{v_{i}} \Bra{v} P \Ket{v} \Bra{v_{i}} F} \\
& = \sum_{i = 1}^{n} \zeta \bullet \Ket{v_{i}} \Bra{v} P \Ket{v} \Bra{v_{i}} F \\
& = \sum_{i = 1}^{n} \SqBr{\Br{\zeta \bullet \Ket{v_{i}} \Bra{v}} \bullet P} \bullet \Ket{v} \Bra{v_{i}} F.
\end{align*}
As $ \Ket{v_{i}} \Bra{v} $ and $ \Ket{v} \Bra{v_{i}} F $ are finite-rank operators on $ \H $ for all $ i \in \SqBr{n} $, they belong to $ \Comp{\H} $. It follows that $ \zeta \bullet \Ket{v_{i}} \Bra{v} \in \Y $ for all $ i \in \SqBr{n} $ because $ \Y $ is a $ \Comp{\H} $-submodule of $ \X $, so $ \zeta \bullet F $ belongs to the $ \Comp{\H} $-linear span of $ \Y \bullet P $, as claimed.

As $ F $ is an arbitrary finite-rank operator on $ \H $, we see that $ \zeta \bullet T $ is in the closed $ \Comp{\H} $-linear span of $ \Y \bullet P $ in $ \X $ for any $ T \in \Comp{\H} $, as $ T $ is the limit in $ \Comp{\H} $ of some sequence of finite-rank operators on $ \H $.

Let $ \Seq{E_{i}}{i \in I} $ be an approximate identity in $ \Comp{\H} $. By the argument above, $ \zeta \bullet E_{i} $ is in the closed $ \Comp{\H} $-linear span of $ \Y \bullet P $ in $ \X $ for all $ i \in I $. As $ \Seq{\zeta \bullet E_{i}}{i \in I} $ converges to $ \zeta $, it follows that $ \zeta $ is in the closed $ \Comp{\H} $-linear span of $ \Y \bullet P $ in $ \X $. As $ \zeta \in \Y $ is arbitrary, $ \Y $ is thus contained in the closed $ \Comp{\H} $-linear span of $ \Y \bullet P $ in $ \X $.

Finally, by the definition of closure, $ \Cl{\Y}{\X} $ is contained in the closed $ \Comp{\H} $-linear span of $ \Y \bullet P $ in $ \X $.
\end{proof}


The next theorem is the main result of \cite{BaGu2}, and it explains why Hilbert $ \Comp{\H} $-modules behave like Hilbert spaces. It says that the $ C^{\ast} $-algebra of adjointable operators on a Hilbert $ \Comp{\H} $-module $ \X $ is isomorphic to the $ C^{\ast} $-algebra of bounded operators on the Hilbert space $ \X \bullet P $, for any rank-one projection $ P $ on $ \H $. At first sight, this seems rather astonishing because $ \X \bullet P $ is generally a much smaller space that $ \X $ itself, and it would be hard to imagine why it should have much to say about $ \X $. However, having seen in \autoref{The Reduction of a Submodule of a Hilbert K(H)-Module by a Rank-One Projection Generates the Submodule} that $ \X \bullet P $ generates a dense submodule of $ \X $, one can start to understand why the theorem holds.

The proof given in \cite{BaGu2} relies on concepts from an earlier paper \cite{BaGu1}, but the proof that we give here is very direct and only depends on the previous definitions and results of this section.



\begin{Thm} \label{The Bakic-Guljas Theorem}
Let $ \X $ be a non-trivial Hilbert $ \Comp{\H} $-module, and $ P $ a rank-one projection on $ \H $. Then $ \X \bullet P $ is an invariant subspace for each $ T \in \Adj{\X} $, and the map
$$
\Map{\Adj{\X}}{\Bdd{\X \bullet P}}{T}{T|_{\X \bullet P}}
$$
is a $ C^{\ast} $-isomorphism, where $ \X \bullet P $ is viewed as a Hilbert space. Furthermore, the restriction of this map to $ \Comp{\X} $ yields a $ C^{\ast} $-isomorphism from $ \Comp{\X} $ to $ \Comp{\X \bullet P} $.
\end{Thm}

\begin{proof}
For each $ T \in \Adj{\X} $, the $ \Comp{\H} $-linearity of $ T $ implies that $ \X \bullet P $ is an invariant subspace of $ T $.

It is easy to check that $ \Map{\Adj{\X}}{\Bdd{\X \bullet P}}{T}{T|_{\X \bullet P}} $ is at least a $ C^{\ast} $-homomorphism. To see that it is injective, let $ S,T \in \Adj{\X} $ satisfy $ S|_{\X \bullet P} = T|_{\X \bullet P} $. Then by the $ \Comp{\H} $-linearity and continuity of both $ S $ and $ T $, they must agree on the closed $ \Comp{\H} $-linear span of $ \X \bullet P $, which is equal to $ \X $ by \autoref{The Reduction of a Submodule of a Hilbert K(H)-Module by a Rank-One Projection Generates the Submodule}. Hence, $ S = T $.

Surjectivity is trickier to prove. Let $ L \in \Bdd{\X \bullet P} $, and let $ \Seq{\varepsilon_{i}}{i \in I} $ be an orthonormal basis of $ \X \bullet P $. For each $ J \in \Fin{I} $, let $ \ds T_{J} \df \sum_{i \in J} \Theta_{\func{L}{\varepsilon_{i}},\varepsilon_{i}} \in \Comp{\X} $ and $ \ds L_{J} \df \sum_{i \in J} \Ket{\func{L}{\varepsilon_{i}}} \Bra{\varepsilon_{i}} \in \Comp{\X \bullet P} $; then for all $ \zeta \in \X \bullet P $,
\begin{align*}
    \func{T_{J}}{\zeta}
& = \sum_{i \in J} \func{\Theta_{\func{L}{\varepsilon_{i}},\varepsilon_{i}}}{\zeta} \\
& = \sum_{i \in J} \func{L}{\varepsilon_{i}} \bullet \Inner{\varepsilon_{i}}{\zeta}_{\X} \\
& = \sum_{i \in J} \func{L}{\varepsilon_{i}} \bullet \Br{P \Inner{\varepsilon_{i}}{\zeta}_{\X} P} \\
& = \sum_{i \in J} \func{L}{\varepsilon_{i}} \bullet \Br{\Inner{\varepsilon_{i}}{\zeta}_{\X \bullet P} \cdot P} \\
& = \sum_{i \in J} \Inner{\varepsilon_{i}}{\zeta}_{\X \bullet P} \cdot \func{L}{\varepsilon_{i}} \\
& = \sum_{i \in J} \func{\Ket{\func{L}{\varepsilon_{i}}} \Bra{\varepsilon_{i}}}{\zeta} \\
& = \func{L_{J}}{\zeta},
\end{align*}
which implies that $ T_{J}|_{\X \bullet P} = L_{J} $. Next, for all $ \zeta \in \X \bullet P $,
$$
  \func{L}{\zeta}
= \func{L}{\sum_{i \in I} \Inner{\varepsilon_{i}}{\zeta}_{\X \bullet P} \cdot \varepsilon_{i}}
= \sum_{i \in I} \Inner{\varepsilon_{i}}{\zeta}_{\X \bullet P} \cdot \func{L}{\varepsilon_{i}}
= \sum_{i \in I} \func{\Ket{\func{L}{\varepsilon_{i}}} \Bra{\varepsilon_{i}}}{\zeta},
$$
so $ \Seq{L_{J}}{J \in \Fin{I}} $ is a net (partially ordered by $ \subseteq $) that strongly converges to $ L $. Also, for all $ J \in \Fin{I} $,
\begin{align*}
       \Norm{L_{J}}_{\Bdd{\X \bullet P}}
& =    \func{\sup}{\Set{\Norm{\func{L_{J}}{\zeta}}_{\X \bullet P}}{\zeta \in \X \bullet P ~ \text{and} ~ \Norm{\zeta}_{\X \bullet P} \leq 1}}
       \\
& =    \func{\sup}
            {
            \Set{\Norm{\func{L_{J}}{\zeta}}_{\X \bullet P}}
                {\zeta \in \Span{\SSet{\varepsilon_{i}}_{i \in J}} ~ \text{and} ~ \Norm{\zeta}_{\X \bullet P} \leq 1}
            } \\
& =    \func{\sup}
            {
            \Set{\Norm{\func{L}{\zeta}}_{\X \bullet P}}
                {\zeta \in \Span{\SSet{\varepsilon_{i}}_{i \in J}} ~ \text{and} ~ \Norm{\zeta}_{\X \bullet P} \leq 1}
            } \\
& \leq \func{\sup}{\Set{\Norm{\func{L}{\zeta}}_{\X \bullet P}}{\zeta \in \X \bullet P ~ \text{and} ~ \Norm{\zeta}_{\X \bullet P} \leq 1}} \\
& =    \Norm{L}_{\Bdd{\X \bullet P}},
\end{align*}
which gives, by the first part, $ \Norm{T_{J}}_{\Adj{\X}} = \Norm{L_{J}}_{\Bdd{\X \bullet P}} \leq \Norm{L}_{\Bdd{\X \bullet P}} $. Notice now that $ \Seq{T_{J}}{J \in \Fin{I}} $ is a norm-bounded net in $ \Adj{\X} $ that strongly converges on the $ \Comp{\H} $-linear span of $ \X \bullet P $, which is dense in $ \X $. By an $ \dfrac{\epsilon}{3} $-argument, $ \Seq{T_{J}}{J \in \Fin{I}} $ strongly converges everywhere to some $ T \in \Bdd{X} $. Likewise, $ \Seq{T_{J}^{\ast}}{J \in \Fin{I}} $ strongly converges everywhere to some $ S \in \Bdd{\X} $. As $ \Inner{\func{T_{J}}{\zeta}}{\eta}_{\X} = \Inner{\zeta}{\func{T_{J}^{\ast}}{\eta}}_{\X} $ for all $ \zeta,\eta \in \X $ and $ J \in \Fin{I} $, taking limits yields $ \Inner{\func{T}{\zeta}}{\eta}_{\X} = \Inner{\zeta}{\func{S}{\eta}}_{\X} $. Therefore, $ T \in \Adj{\X} $ and $ T|_{\X \bullet P} = L $, which finishes our proof that the restriction map is a $ C^{\ast} $-isomorphism from $ \Adj{\X} $ to $ \Bdd{\X \bullet P} $.

For the final part of the proof, observe for all $ \zeta,\eta \in \X \bullet P $, $ S,T \in \Comp{\H} $, and $ \xi \in \X \bullet P $ that
\begin{align*}
    \func{\Theta_{\zeta \bullet S,\eta \bullet T}}{\xi}
& = \Br{\zeta \bullet S} \bullet \Inner{\eta \bullet T}{\xi}_{\X} \\
& = \Br{\zeta \bullet S} \bullet \Inner{\eta \bullet P T}{\xi}_{\X} \\
& = \Br{\zeta \bullet S} \bullet \Br{T^{\ast} P \Inner{\eta}{\xi}_{\X}} \\
& = \Br{\zeta \bullet S T^{\ast} P} \bullet \Inner{\eta}{\xi}_{\X} \\
& = \Br{\zeta \bullet S T^{\ast} P} \bullet \Br{P \Inner{\eta}{\xi}_{\X} P} \\
& = \Br{\zeta \bullet S T^{\ast} P} \bullet \Br{\Inner{\eta}{\xi}_{\X \bullet P} \cdot P} \\
& = \Inner{\eta}{\xi}_{\X \bullet P} \cdot \Br{\zeta \bullet S T^{\ast} P} \\
& = \func{\Ket{\zeta \bullet S T^{\ast} P} \Bra{\eta}}{\xi},
\end{align*}
which means that $ \Theta_{\zeta \bullet S,\eta \bullet T}|_{\X \bullet P} \in \Comp{\X \bullet P} $. Let $ \zeta,\eta \in \X $. Then by \autoref{The Reduction of a Submodule of a Hilbert K(H)-Module by a Rank-One Projection Generates the Submodule}, we can find
$$
\zeta_{1},\ldots,\zeta_{m},\eta_{1},\ldots,\eta_{n} \in \X \bullet P
\quad \text{and} \quad
S_{1},\ldots,S_{m},T_{1},\ldots,T_{n} \in \Comp{\H}
$$
so that $ \ds \sum_{i = 1}^{m} \zeta_{i} \bullet S_{i} $ and $ \ds \sum_{j = 1}^{n} \eta_{j} \bullet T_{j} $ are arbitrarily close to $ \zeta $ and $ \eta $, respectively, which ensures that
$$
  \Theta_{\sum_{i = 1}^{m} \zeta_{i} \bullet S_{i},\sum_{j = 1}^{n} \eta_{j} \bullet T_{j}}
= \sum_{i = 1}^{m} \sum_{j = 1}^{n} \Theta_{\zeta_{i} \bullet S_{i},\eta_{j} \bullet T_{j}}
$$
is arbitrarily close to $ \Theta_{\zeta,\eta} $ in $ \Comp{\X} $. Hence, by the continuity of the restriction $ C^{\ast} $-isomorphism,
$$
  \sum_{i = 1}^{m} \sum_{j = 1}^{n} \Theta_{\zeta_{i} \bullet S_{i},\eta_{j} \bullet T_{j}}|_{\X \bullet P}
= \sum_{i = 1}^{m} \sum_{j = 1}^{n} \Ket{\zeta_{i} \bullet S_{i} T_{j}^{\ast} P} \Bra{\eta_{j}}
$$
is arbitrarily close to $ \Theta_{\zeta,\eta}|_{\X \bullet P} $, which says that $ \Theta_{\zeta,\eta}|_{\X \bullet P} \in \Comp{\X \bullet P} $. Therefore, the image of $ \Comp{\X} $ under the restriction $ C^{\ast} $-isomorphism is a non-trivial ideal in $ \Comp{\X \bullet P} $. As $ \Comp{\X \bullet P} $ is a simple $ C^{\ast} $-algebra, this image is precisely $ \Comp{\X \bullet P} $. The proof is now complete.
\end{proof}


Using \autoref{The Bakic-Guljas Theorem}, we can show that every closed submodule of a Hilbert $ \Comp{\H} $-module has an orthogonal complement. The complementability of Hilbert $ \Comp{\H} $-modules has been known for a long while (\cite{Mag}), but \autoref{The Bakic-Guljas Theorem} appears to provide an expedient proof.



\begin{Thm} \label{The Complementability of Hilbert K(H)-Modules}
Let $ \Y $ be a closed submodule of a Hilbert $ \Comp{\H} $-module $ \X $. Then $ \X = \Y \oplus \Y^{\perp} $.
\end{Thm}

\begin{proof}
Let $ P $ be a rank-one projection on $ \H $. Then $ \Y \bullet P $ is a closed subspace of $ \X \bullet P $. By \autoref{The Bakic-Guljas Theorem}, there is a projection $ Q \in \Adj{\X} $ such that $ Q|_{\X \bullet P} = \Proj{\X \bullet P}{\Y \bullet P} $. If we can show that $ \Range{Q} = \Y $, then we are done, for a closed submodule of a Hilbert $ C^{\ast} $-module is complementable if it is the range of an adjointable operator (Corollary 15.3.9 of \cite{We}). Indeed, as $ Q $ is a projection, $ \Range{Q} $ is a closed submodule of $ \X $, and
$$
  \Range{Q} \bullet P
= \Set{\func{Q}{\zeta} \bullet P}{\zeta \in \X}
= \Set{\func{Q}{\zeta \bullet P}}{\zeta \in \X}
= \Range{Q|_{\X \bullet P}}
= \Range{\Proj{\X \bullet P}{\Y \bullet P}}
= \Y \bullet P,
$$
so $ \Range{Q} = \Y $ by \autoref{The Reduction of a Submodule of a Hilbert K(H)-Module by a Rank-One Projection Generates the Submodule}.
\end{proof}



\begin{Lem} \label{The Existence of an Element of a K(H)-Module Whose Inner Product with Itself Is a Rank-One Projection}
Let $ \X $ be a non-trivial Hilbert $ \Comp{\H} $-module, and $ P $ a rank-one projection on $ \H $. Then there is a $ \zeta \in \X $ such that $ \Inner{\zeta}{\zeta}_{\X} = P $.
\end{Lem}

\begin{proof}
By \autoref{The Hilbert Space Coming from a Rank-One Projection}, we can find a non-zero $ \eta \in \X \bullet P $, so $ \Inner{\eta}{\eta}_{\X \bullet P} > 0 $. As
$$
  \Inner{\eta}{\eta}_{\X}
= \Inner{\eta \bullet P}{\eta \bullet P}_{\X}
= P \Inner{\eta}{\eta}_{\X} P
= \Inner{\eta}{\eta}_{\X \bullet P} \cdot P,
$$
we see that $ \zeta \df \dfrac{1}{\sqrt{\Inner{\eta}{\eta}_{\X \bullet P}}} \cdot \eta $ satisfies $ \Inner{\zeta}{\zeta}_{\X} = P $.
\end{proof}



\begin{Def} \label{Irreducibility}
Let $ A $ be a non-trivial $ C^{\ast} $-algebra, and $ \X $ a non-trivial Hilbert $ A $-module. We say that $ \Comp{\X} $ \emph{acts irreducibly} on $ \X $ if and only if the only closed submodules of $ \X $ that are invariant under the left action of $ \Comp{\X} $ are $ \SSet{0_{\X}} $ and $ \X $.
\end{Def}



\begin{Prop} \label{K(X) Acts Irreducibly on X}
Let $ \X $ be a non-trivial Hilbert $ \Comp{\H} $-module. Then $ \Comp{\X} $ acts irreducibly on $ \X $.
\end{Prop}

\begin{proof}
Let $ P $ be a rank-one projection on $ \H $. By \autoref{The Hilbert Space Coming from a Rank-One Projection}, $ \X \bullet P $ is a non-trivial Hilbert space. Suppose that $ \Comp{\X} $ does not act irreducibly on $ \X $, i.e., there is a non-trivial proper closed submodule $ \Y $ of $ \X $ that is invariant under the left action of $ \Comp{\X} $. By \autoref{The Hilbert Space Coming from a Rank-One Projection} again, $ \Y \bullet P $ is a non-trivial closed subspace of $ \X \bullet P $. We will then have a contradiction if we can show that $ \Y \bullet P $ is a non-trivial and proper closed subspace of $ \X \bullet P $ that is invariant under the left action of $ \Comp{\X \bullet P} $ (note that the $ C^{\ast} $-algebra of compact operators on a non-trivial Hilbert space must act irreducibly on the Hilbert space).

If $ \Y \bullet P $ were not a proper subspace of $ \X \bullet P $, then $ \Y \bullet P = \X \bullet P $, so the closed $ \Comp{\H} $-linear span of $ \Y \bullet P $ in $ \X $ would be that of $ \X \bullet P $ in $ \X $. By \autoref{The Reduction of a Submodule of a Hilbert K(H)-Module by a Rank-One Projection Generates the Submodule}, we would obtain $ \Y = \X $, which violates the earlier assumption that $ \Y $ is a proper submodule of $ \X $. Hence, $ \Y \bullet P \subsetneq \X \bullet P $.

Let $ T \in \Comp{\X \bullet P} $. Then $ T = S|_{\X \bullet P} $ for some $ S \in \Comp{\X} $, so
$$
          \Im{T}{\Y \bullet P}
=         \Im{S}{\Y \bullet P}
=         \Im{S}{\Y} \bullet P
\subseteq \Y \bullet P,
$$
proving that $ \Y \bullet P $ is invariant under the left action of $ \Comp{\X \bullet P} $. This produces the desired contradiction, which completes the proof of this proposition.
\end{proof}



\begin{Prop} \label{A Decomposition Result for Non-Degenerate *-Representations of K(X)}
Let $ \X $ and $ \Y $ be non-trivial Hilbert $ \Comp{\H} $-modules. If $ \Phi $ is a non-degenerate $ \ast $-representation of $ \Comp{\X} $ on $ \Y $, then $ \Pair{\Y}{\Phi} $ is unitarily equivalent to a direct sum of copies of $ \Pair{\X}{i_{\Comp{\X} \into \Adj{\X}}} $.
\end{Prop}

\begin{proof}
Our proof is an adaptation of Arveson's proof of Theorem 1.4.4 in \cite{Ar}. Fix a rank-one projection $ P $ on $ \H $, and consider
$$
\Psi:
\Comp{\X \bullet P} \stackrel{\cong}{\longrightarrow}
\Comp{\X}           \stackrel{\Phi}{\longrightarrow}
\Adj{\Y}            \stackrel{\cong}{\longrightarrow}
\Bdd{\Y \bullet P},
$$
where $ \Comp{\X \bullet P} \stackrel{\cong}{\longrightarrow} \Comp{\X} $ and $ \Adj{\Y} \stackrel{\cong}{\longrightarrow} \Bdd{\Y \bullet P} $ come from \autoref{The Bakic-Guljas Theorem}. By definition, the non-degeneracy of $ \Pair{\Phi}{\Y} $ means that
$$
\Y = \Cl{\Span{\Set{\FUNC{\func{\Phi}{T}}{\zeta}}{T \in \Comp{\X} ~ \text{and} ~ \zeta \in \Y}}}{\Y},
$$
which yields
\begin{align*}
            \Y \bullet P
& =         \Cl{\Span{\Set{\FUNC{\func{\Phi}{T}}{\zeta}}{T \in \Comp{\X} ~ \text{and} ~ \zeta \in \Y}}}{\Y} \bullet P \\
& \subseteq \Cl{\Span{\Set{\FUNC{\func{\Phi}{T}}{\zeta}}{T \in \Comp{\X} ~ \text{and} ~ \zeta \in \Y}} \bullet P}{\Y} \\
& =         \Cl{\Span{\Set{\FUNC{\func{\Phi}{T}}{\zeta \bullet P}}{T \in \Comp{\X} ~ \text{and} ~ \zeta \in \Y}}}{\Y} \\
& =         \Cl{\Span{\Set{\FUNC{\func{\Phi}{T}|_{\Y \bullet P}}{\eta}}{T \in \Comp{\X} ~ \text{and} ~ \eta \in \Y \bullet P}}}{\Y} \\
& =         \Cl{\Span{\Set{\FUNC{\func{\Psi}{S}}{\eta}}{S \in \Comp{\X \bullet P} ~ \text{and} ~ \eta \in \Y \bullet P}}}{\Y} \\
& =         \Cl{\Span{\Set{\FUNC{\func{\Psi}{S}}{\eta}}{S \in \Comp{\X \bullet P} ~ \text{and} ~ \eta \in \Y \bullet P}}}{\Y \bullet P} \qquad
            \Br{\text{By \autoref{The Hilbert Space Coming from a Rank-One Projection}.}} \\
& \subseteq \Y \bullet P.
\end{align*}
Hence, $ \Psi $ is a non-degenerate $ \ast $-representation of $ \Comp{\X \bullet P}$ on the Hilbert space $ \Y \bullet P $, and so the first part of Arveson's proof says that there is a rank-one projection $ Q \in \Comp{\X \bullet P} $ such that $ \func{\Psi}{Q} \neq 0_{\Comp{\Y \bullet P}} $.

By \autoref{The Bakic-Guljas Theorem}, there is a projection $ E \in \Comp{\X} $ such that $ Q = E|_{\X \bullet P} $. As $ \func{\Psi}{Q} \neq 0_{\Comp{\Y \bullet P}} $, it must be that $ \func{\Phi}{E} \neq 0_{\Adj{\Y}} $, so $ E \neq 0_{\Comp{\X}} $. By \autoref{The Linear Functional Associated To a Rank-One Projection}, there exists a linear functional $ f_{P}: \Comp{\X \bullet P} \to \C $ satisfying
$$
\forall S \in \Comp{\X \bullet P}: \qquad
Q S Q = \func{f_{P}}{S} \cdot Q.
$$
Define a linear functional $ g: \Comp{\X} \to \C $ by $ \func{g}{T} \df \func{f_{P}}{T|_{\X \bullet P}} $ for all $ T \in \Comp{\X} $; then
$$
  E T E|_{\X \bullet P}
= \Br{E|_{\X \bullet P}} \Br{T|_{\X \bullet P}} \Br{E|_{\X \bullet P}}
= Q \Br{T|_{\X \bullet P}} Q
= \func{f_{P}}{T|_{\X \bullet P}} \cdot Q
= \func{f_{P}}{T|_{\X \bullet P}} \cdot E|_{\X \bullet P}
= \SqBr{\func{g}{T} \cdot E}|_{\X \bullet P}.
$$
By \autoref{The Bakic-Guljas Theorem} again, we may conclude that $ E T E = \func{g}{T} \cdot E $ for all $ T \in \Comp{\X} $.

Consider the $ \Comp{\H} $-submodule $ \Im{E}{\X} $ of $ \X $, which is non-trivial as $ E \neq 0_{\Comp{\X}} $, and closed as $ E $ is a projection. Similarly, $ \Im{\func{\Phi}{E}}{\Y} $ is a non-trivial closed $ \Comp{\H} $-submodule of $ \Y $. Hence, by \autoref{The Existence of an Element of a K(H)-Module Whose Inner Product with Itself Is a Rank-One Projection}, there exist $ \zeta \in \Im{E}{\X} $ and $ \eta \in \Im{\func{\Phi}{E}}{\Y} $ such that $ \Inner{\zeta}{\zeta}_{\X} = P = \Inner{\eta}{\eta}_{\Y} $. We now claim that the map
$$
U:
\Span{\Set{\func{T}{\zeta \bullet A} \in \X}{T \in \Comp{\X}, ~ A \in \Comp{\H}}} \to
\Span{\Set{\FUNC{\func{\Phi}{T}}{\eta \bullet A} \in \Y}{T \in \Comp{\X}, ~ A \in \Comp{\H}}}
$$
defined by
$$
\sum_{i = 1}^{n} \func{T_{i}}{\zeta \bullet A_{i}} \mapsto \sum_{i = 1}^{n} \FUNC{\func{\Phi}{T_{i}}}{\eta \bullet A_{i}}
$$
for all $ T_{1},\ldots,T_{n} \in \Comp{\X} $ and $ A_{1},\ldots,A_{n} \in \Comp{\H} $ is well-defined by virtue of being an isometry. Indeed,
\begin{align*}
  & ~ \Norm{
           \Inner{\sum_{i = 1}^{n} \FUNC{\func{\Phi}{T_{i}}}{\eta \bullet A_{i}}}
                 {\sum_{j = 1}^{n} \FUNC{\func{\Phi}{T_{j}}}{\eta \bullet A_{j}}}_{\Y}
           }_{\Comp{\H}} \\
= & ~ \Norm{
           \sum_{i,j = 1}^{n} \Inner{\FUNC{\func{\Phi}{T_{i}}}{\eta \bullet A_i}}{\FUNC{\func{\Phi}{T_{j}}}{\eta \bullet A_{j}}}_{\Y}
           }_{\Comp{\H}} \\
= & ~ \Norm{
           \sum_{i,j = 1}^{n}
           \Inner{\FUNC{\func{\Phi}{T_{i}}}{\FUNC{\func{\Phi}{E}}{\eta} \bullet A_{i}}}
                 {\FUNC{\func{\Phi}{T_{j}}}{\FUNC{\func{\Phi}{E}}{\eta} \bullet A_{j}}}_{\Y}
           }_{\Comp{\H}} \qquad \Br{\text{As $ \eta \in \Im{\SqBr{\func{\Phi}{E}}}{\Y} $.}} \\
= & ~ \Norm{
           \sum_{i,j = 1}^{n} \Inner{\FUNC{\func{\Phi}{T_{i} E}}{\eta \bullet A_{i}}}{\FUNC{\func{\Phi}{T_{j} E}}{\eta \bullet A_{j}}}_{\Y}
           }_{\Comp{\H}} \\
= & ~ \Norm{
           \sum_{i,j = 1}^{n} \Inner{\FUNC{\func{\Phi}{T_{j} E}^{\ast} \func{\Phi}{T_{i} E}}{\eta \bullet A_{i}}}{\eta \bullet A_{j}}_{\Y}
           }_{\Comp{\H}} \\
= & ~ \Norm{\sum_{i,j = 1}^{n} \Inner{\FUNC{\func{\Phi}{E T_{j}^{\ast} T_{i} E}}{\eta \bullet A_{i}}}{\eta \bullet A_{j}}_{\Y}}_{\Comp{\H}} \\
= & ~ \Norm{
           \sum_{i,j = 1}^{n} \Inner{\FUNC{\func{\Phi}{\func{g}{T_{j}^{\ast} T_{i}} \cdot E}}{\eta \bullet A_{i}}}{\eta \bullet A_{j}}_{\Y}
           }_{\Comp{\H}} \\
= & ~ \Norm{
           \sum_{i,j = 1}^{n}
           \overline{\func{g}{T_{j}^{\ast} T_{i}}} \cdot \Inner{\FUNC{\func{\Phi}{E}}{\eta \bullet A_{i}}}{\eta \bullet A_{j}}_{\Y}
           }_{\Comp{\H}} \\
= & ~ \Norm{\sum_{i,j = 1}^{n} \overline{\func{g}{T_{j}^{\ast} T_{i}}} \cdot \Inner{\eta \bullet A_{i}}{\eta \bullet A_{j}}_{\Y}}_{\Comp{\H}}
      \\
= & ~ \Norm{\sum_{i,j = 1}^{n} \overline{\func{g}{T_{j}^{\ast} T_{i}}} \cdot A_{i}^{\ast} \Inner{\eta}{\eta}_{\Y} A_{j}}_{\Comp{\H}} \\
= & ~ \Norm{\sum_{i,j = 1}^{n} \overline{\func{g}{T_{j}^{\ast} T_{i}}} \cdot A_{i}^{\ast} P A_{j}}_{\Comp{\H}}, \qquad
      \Br{\text{As $ \Inner{\eta}{\eta}_{\Y} = P $.}}
\end{align*}
and a nearly-identical computation using $ \Inner{\zeta}{\zeta}_{\X} = P $ also yields
$$
  \Norm{\Inner{\sum_{i = 1}^{n} \func{T_{i}}{\zeta \bullet A_i}}{\sum_{j = 1}^{n} \func{T_{j}}{\zeta \bullet A_{j}}}_{\Y}}_{\Comp{\H}}
= \Norm{\sum_{i,j = 1}^{n} \overline{\func{g}{T_{j}^{\ast} T_{i}}} \cdot A_{i}^{\ast} P A_{j}}_{\Comp{\H}}.
$$
Therefore, $ U $ is a surjective isometry. By continuity, $ U $ extends to a surjective isometry $ U : \X' \to \Y' $, where
$$
\X' \df \Cl{\Span{\Set{\func{T}{\zeta \bullet A}}{T \in \Comp{\X}, ~ A \in \Comp{\H}}}}{\X}
$$
and
$$
\Y' \df \Cl{\Span{\Set{\FUNC{\func{\Phi}{T}}{\eta \bullet A}}{T \in \Comp{\X} ~ A \in \Comp{\H}}}}{\Y}.
$$
Note that $ \X' $ is a $ \Comp{\X} $-invariant closed submodule of $ \X $, and is non-trivial as $ \zeta \in \X' $. Also, $ \Y' $ is a $ \Im{\Phi}{\Comp{\X}} $-invariant closed submodule of $ \Y $, and is non-trivial as $ \eta \in \Y' $. Hence, $ \X' = \X $ by \autoref{K(X) Acts Irreducibly on X}, so $ U: \X \to \Y' $ is a surjective isometry that, moreover, is $ \Comp{\H} $-linear. We may thus apply Theorem 3.5(i) of \cite{La} to deduce that $ U \in \Unitary{\X,\Y'} $.

Next, we claim that $ U T = \func{\Phi}{T}|_{\Y'} U $ for all $ T \in \Comp{\X} $. Fix $ T \in \Comp{\X} $. Then for all $ T_{1},\ldots,T_{n} \in \Comp{\X} $ and $ A_{1},\ldots,A_{n} \in \Comp{\H} $, we have
\begin{align*}
    \Func{U T}{\sum_{i = 1}^{n} \func{T_{i}}{\xi \bullet A_{i}}}
& = \func{U}{\sum_{i = 1}^{n} \func{T T_{i}}{\xi \bullet A_{i}}} \\
& = \sum_{i = 1}^{n} \FUNC{\func{\Phi}{T T_{i}}}{\eta \bullet A_{i}} \\
& = \sum_{i = 1}^{n} \FUNC{\func{\Phi}{T} \func{\Phi}{T_{i}}}{\eta \bullet A_{i}} \\
& = \FUNC{\func{\Phi}{T}}{\sum_{i = 1}^{n} \FUNC{\func{\Phi}{T_{i}}}{\eta \bullet A_{i}}} \\
& = \FUNC{\func{\Phi}{T} U}{\sum_{i = 1}^{n} \func{T_{i}}{\xi \bullet A_{i}}}.
\end{align*}
By the density of $ \Span{\Set{\func{T}{\xi \bullet A}}{T \in \Comp{\X}, ~ a \in \Comp{\H}}} $ in $ \X $, we obtain $ U T = \func{\Phi}{T}|_{\Y'} U $ as expected.

Now, define a poset $ \Pair{\mathfrak{P}}{\sqsubseteq} $ with the following properties:
\begin{itemize}
\item
$ \mathcal{S} $ is an element of $ \mathfrak{P} $ if and only if the following hold:
\begin{itemize}
\item
$ \mathcal{S} $ consists of pairs of the form $ \Pair{\Z}{V} $, where $ \Z $ is a non-trivial $ \Im{\Phi}{\Comp{\X}} $-invariant closed submodule of $ \Y $, and $ V \in \Unitary{\X,\Z} $ with $ T = V^{- 1} \SqBr{\func{\Phi}{T}|_{\Z}} V $ for all $ T \in \Comp{\H} $.

\item
If $ \Pair{\Z_{1}}{V_{1}} $ and $ \Pair{\Z_{2}}{V_{2}} $ are distinct elements of $ \mathcal{S} $, then $ \Z_{1} \perp \Z_{2} $.
\end{itemize}

\item
For all $ \mathcal{S}_{1},\mathcal{S}_{2} \in \mathfrak{P} $, we have $ \mathcal{S}_{1} \sqsubseteq \mathcal{S}_{2} $ if and only if $ \mathcal{S}_{1} \subseteq \mathcal{S}_{2} $.
\end{itemize}
If $ \mathfrak{C} $ is a chain in $ \Pair{\mathfrak{P}}{\sqsubseteq} $, then $ \bigcup \mathfrak{C} $ is an upper bound for $ \mathfrak{C} $ in $ \Pair{\mathfrak{P}}{\sqsubseteq} $, so by Zorn's Lemma, there exists a maximal element $ \mathcal{M} $ of $ \Pair{\mathfrak{P}}{\sqsubseteq} $. We claim that $ \ds \Y = \bigoplus_{\Pair{\Z}{V} \in \mathcal{M}} \Z $, where the direct sum is internal. If this were not true, then $ \ds \bigoplus_{\Pair{\Z}{V} \in \mathcal{M}} \Z \subsetneq \Y $. Letting $ \ds \Z' \df \SqBr{\bigoplus_{\Pair{\Z}{V} \in \mathcal{M}} \Z}^{\perp} $, \autoref{The Complementability of Hilbert K(H)-Modules} says that $ \Z' $ is a non-trivial closed submodule of $ \Y $. A routine verification reveals that $ \Im{\func{\Phi}{T}}{\Z'} \subseteq \Z' $ for each $ T \in \Comp{\H} $, and that
$$
\Map{\Comp{\X}}{\Adj{\Z'}}{T}{\func{\Phi}{T}|_{\Z'}}
$$
is a non-degenerate $ \ast $-representation of $ \Comp{\X} $ on $ \Z' $. We may thus apply the first part of the proof to $ \Z' $ to obtain $ \Pair{\Z''}{W} $, where
\begin{itemize}
\item
$ \Z'' $ is a non-trivial $ \Im{\Phi}{\Comp{\X}} $-invariant closed submodule of $ \Z' $ (and hence of $ \Y $), and

\item
$ W \in \Unitary{\X,\Z''} $ with $ T = W^{- 1} \SqBr{\func{\Phi}{T}|_{\Z''}} W $ for all $ T \in \Comp{\H} $.
\end{itemize}
As $ \mathcal{M} \subsetneq \SSet{\Pair{\Z''}{W}} \cup \mathcal{M} \in \mathcal{P} $, this contradicts the maximality of $ \mathcal{M} $. Therefore, $ \ds \Y = \bigoplus_{\Pair{\Z}{V} \in \mathcal{M}} \Z $ indeed, so
$$
\forall T \in \Comp{\H}: \qquad
  \func{\Phi}{T}
= \oplus_{\Pair{\Z}{V} \in \mathcal{M}} \func{\Phi}{T}|_{\Z}
= \oplus_{\Pair{\Z}{V} \in \mathcal{M}} V T V^{- 1}
= \SqBr{\oplus_{\Pair{\Z}{V} \in \mathcal{M}} V} T \SqBr{\oplus_{\Pair{\Z}{V} \in \mathcal{M}} V^{- 1}}.
$$
The proof is finally complete.
\end{proof}




\section{The Covariant Stone-von Neumann Theorem}


In \cite{Ri1}, Marc Rieffel applied a special instance of the next theorem --- Green's Imprimitivity Theorem --- to derive the classical Stone-von Neumann Theorem. According to him, the classical Stone-von Neumann Theorem is really a statement about the Morita equivalence of the $ C^{\ast} $-algebra $ \C $ with the crossed product $ \Cstar{G,\Co{G},\lt} $. This gives us a more algebraic way of seeing things, and it is precisely this point of view that guided our search for the covariant Stone-von Neumann Theorem in the beginning. As we are dealing with Hilbert $ C^{\ast} $-modules instead of just Hilbert spaces, the full strength of Green's Imprimitivity Theorem is required.



\begin{Thm}[Green's Imprimitivity Theorem] \label{Green's Imprimitivity Theorem}
Let $ \Trip{G}{A}{\alpha} $ be a $ C^{\ast} $-dynamical system. Then $ \LL{G}{A}{\alpha} $ is a $ \Pair{\Cstar{G,\Co{G,A},\lt \otimes \alpha}}{A} $-imprimitivity bimodule with the following properties:
\begin{itemize}
\item
If $ \Xi $ denotes the non-degenerate $ \ast $-representation of $ \Co{G,A} $ on $ \LL{G}{A}{\alpha} $ uniquely determined by
$$
\forall g \in \Co{G,A}, ~ \forall \phi \in \Cc{G,A}: \qquad
\FUNC{\func{\Xi}{g}}{\func{q}{\phi}} = \func{q}{g \phi},
$$
then $ \Trip{\LL{G}{A}{\alpha}}{\Xi}{\U} $ is a $ \Trip{G}{\Co{G,A}}{\lt \otimes \alpha} $-covariant modular representation, and the left action of $ \Cstar{G,\Co{G,A},\lt \otimes \alpha} $ on $ \LL{G}{A}{\alpha} $ is $ \overline{\Pi}_{\LL{G}{A}{\alpha},\Xi,\U} $.

\item
The $ \Cstar{G,\Co{G,A},\lt \otimes \alpha} $-valued inner product on $ \LL{G}{A}{\alpha} $ is uniquely determined by
$$
  _{\Cstar{G,\Co{G,A},\lt \otimes \alpha}}\Inner{\func{q}{\phi}}{\func{q}{\psi}}
= \func{\eta_{\Trip{G}{\Co{G,A}}{\lt \otimes \alpha}}}
       {\Map{G}{\Co{G,A}}{x}{\func{\Delta_{G}}{x^{- 1} \bullet} \cdot \func{\phi}{\bullet} \alphArg{x}{\func{\psi}{x^{- 1} \bullet}^{\ast}}}}.
$$
for all $ \phi,\psi \in \Cc{G,A} $.

\item
Both the right $ A $-action and the $ A $-valued inner product on $ \LL{G}{A}{\alpha} $ are precisely the ones that define $ \LL{G}{A}{\alpha} $ as a Hilbert $ A $-module.
\end{itemize}
\end{Thm}


Complete proofs of Green's Imprimitivity Theorem may be found in \cite{Gr,Wi}.


\begin{Prop} \label{The Covariant Pre-Stone-von Neumann Theorem}
Let $ \Trip{G}{A}{\alpha} $ be a $ C^{\ast} $-dynamical system with $ G $ abelian. Recalling the $ \Trip{G}{A}{\alpha} $-Schr\"odinger modular representation $ \Quad{\LL{G}{A}{\alpha}}{\M}{\U}{\V} $, and letting $ \nu $ be any Haar measure on $ \widehat{G} $, we have that
$$
\overline{\Pi}_{\LL{G}{A}{\alpha},\pi^{\LL{G}{A}{\alpha},\M,\V}_{\nu},\U}: \Cstar{G,\Co{G,A},\lt \otimes \alpha} \to \Adj{\LL{G}{A}{\alpha}}
$$
is an injective $ C^{\ast} $-homomorphism and that $ \Range{\Pi_{\LL{G}{A}{\alpha},\pi^{\LL{G}{A}{\alpha},\M,\V}_{\nu},\U}} = \Comp{\LL{G}{A}{\alpha}} $.
\end{Prop}

\begin{proof}
Firstly, we show that $ \pi^{\LL{G}{A}{\alpha},\M,\V}_{\nu} = \Xi $. Let $ f \in \Cc{\widehat{G},A} $ and $ \phi \in \Cc{G,A} $. Then
\begin{align*}
    \FUNC{\func{\pi^{\LL{G}{A}{\alpha},\M,\V}_{\nu}}{\FT{\nu}{A}{f}}}{\func{q}{\phi}}
& = \FUNC{\func{\Pi_{\LL{G}{A}{\alpha},\M,\V}}{f}}{\func{q}{\phi}} \\
& = \Int{\widehat{G}}{\FUNC{\func{\M}{\func{f}{\varphi}} \circ \func{\V}{\varphi}}{\func{q}{\phi}}}{\nu}{\varphi} \\
& = \Int{\widehat{G}}{\func{q}{\func{f}{\varphi} \Br{\varphi \cdot \phi}}}{\nu}{\varphi}.
\end{align*}
The last integral looks like it should be $ \func{q}{\FT{\nu}{A}{f} \phi} $, and indeed it is, but we have to exercise some caution in justifying our guess. By Fubini's Theorem, we have for all $ \psi \in \Cc{G,A} $ that
\begin{align*}
    \Inner{\func{q}{\psi}}{\Int{\widehat{G}}{\func{q}{\func{f}{\varphi} \Br{\varphi \cdot \phi}}}{\nu}{\varphi}}_{\LL{G}{A}{\alpha}}
& = \Int{\widehat{G}}{\Inner{\func{q}{\psi}}{\func{q}{\func{f}{\varphi} \Br{\varphi \cdot \phi}}}_{\LL{G}{A}{\alpha}}}{\nu}{\varphi} \\
& = \Int{\widehat{G}}
        {\SqBr{\Int{G}{\alphArg{x^{- 1}}{\func{\psi}{x}^{\ast} \func{f}{\varphi} \SqBr{\func{\varphi}{x} \cdot \func{\phi}{x}}}}{\mu}{x}}}
        {\nu}{\varphi} \\
& = \Int{\widehat{G}}
        {\SqBr{\Int{G}{\alphArg{x^{- 1}}{\func{\psi}{x}^{\ast} \SqBr{\func{\varphi}{x} \cdot \func{f}{\varphi} \func{\phi}{x}}}}{\mu}{x}}}
        {\nu}{\varphi} \\
& = \Int{G}
        {
        \SqBr{
             \Int{\widehat{G}}
                 {\alphArg{x^{- 1}}{\func{\psi}{x}^{\ast} \SqBr{\func{\varphi}{x} \cdot \func{f}{\varphi} \func{\phi}{x}}}}
                 {\nu}{\varphi}
             }
        }
        {\mu}{x} \\
& = \Int{G}
        {
        \alphArg{x^{- 1}}
                {\func{\psi}{x}^{\ast} \SqBr{\Int{\widehat{G}}{\func{\varphi}{x} \cdot \func{f}{\varphi}}{\nu}{\varphi}} \func{\phi}{x}}
        }
        {\mu}{x} \\
& = \Int{G}{\alphArg{x^{- 1}}{\func{\psi}{x}^{\ast} \FUNC{\FT{\nu}{A}{f}}{x} \func{\phi}{x}}}{\mu}{x} \\
& = \Inner{\func{q}{\psi}}{\func{q}{\FT{\nu}{A}{f} \phi}}_{\LL{G}{A}{\alpha}}.
\end{align*}
This establishes the validity of our guess. Hence, $ \func{\pi^{\LL{G}{A}{\alpha},\M,\V}_{\nu}}{\FT{\nu}{A}{f}} = \func{\Xi}{\FT{\nu}{A}{f}} $ for all $ f \in \Cc{\widehat{G},A} $, and as the range of $ \F{\nu}{A} $ is dense in $ \Co{G,A} $, we conclude that $ \pi^{\LL{G}{A}{\alpha},\M,\V}_{\nu} = \Xi $.

It now follows from Green's Imprimitivity Theorem and Proposition 3.8 of \cite{RaWi} that
$$
\overline{\Pi}_{\LL{G}{A}{\alpha},\pi^{\LL{G}{A}{\alpha},\M,\V}_{\nu},\U}: \Cstar{G,\Co{G,A},\lt \otimes \alpha} \to \Adj{\LL{G}{A}{\alpha}}
$$
is an injective $ C^{\ast} $-homomorphism whose range is $ \Comp{\LL{G}{A}{\alpha}} $.
\end{proof}



\begin{Def}
Let $ \Trip{G}{A}{\alpha} $ and $ \Trip{G}{A}{\beta} $ be $ C^{\ast} $-dynamical systems with $ G $ abelian. Let $ \Quad{\X}{\rho}{R}{S} $ be a $ \Trip{G}{A}{\alpha} $-Heisenberg modular representation; $ \Quad{\Y}{\sigma}{T}{U} $ a $ \Trip{G}{A}{\beta} $-Heisenberg modular representation. We say that $ \Quad{\X}{\rho}{R}{S} $ is \emph{unitarily equivalent} to $ \Quad{\Y}{\sigma}{T}{U} $ if and only if there exists a $ W \in \Unitary{\X,\Y} $ such that
$$
W \func{R}{x} W^{\ast}        = \func{T}{x},      \qquad
W \func{S}{\varphi} W^{\ast} = \func{U}{\varphi}, \qquad
W \func{\rho}{a} W^{\ast}     = \func{\sigma}{a}
$$
for all $ x \in G $, $ \varphi \in \widehat{G} $, and $ a \in A $, in which case we write $ \Quad{\X}{\rho}{R}{S} \sim_{W} \Quad{\Y}{\sigma}{T}{U} $.
\end{Def}



\begin{Def} \label{The von Neumann Uniqueness Property}
Let $ \Trip{G}{A}{\alpha} $ be a $ C^{\ast} $-dynamical system with $ G $ abelian. We say that $ \Trip{G}{A}{\alpha} $ has the \emph{von Neumann Uniqueness Property} if and only if each $ \Trip{G}{A}{\alpha} $-Heisenberg modular representation is unitarily equivalent to a direct sum of copies of the $ \Trip{G}{A}{\alpha} $-Schr\"odinger modular representation.
\end{Def}


We are now ready to prove the main result of this paper.



\begin{Prop}[The Covariant Stone-von Neumann Theorem] \label{The Covariant Stone-von Neumann Theorem}
Every $ C^{\ast} $-dynamical system of the form $ \Trip{G}{\Comp{\H}}{\alpha} $, with $ G $ abelian, has the von Neumann Uniqueness Property.
\end{Prop}

\begin{proof}
Let $ \Quad{\X}{\rho}{R}{S} $ be a $ \Trip{G}{\Comp{\H}}{\alpha} $-Heisenberg modular representation. According to \autoref{An Injective Map from the Class of Heisenberg Modular Representations to the Class of Covariant Modular Representations}, $ \Trip{\X}{\pi^{\X,\rho,S}_{\nu}}{R} $ is a $ \Trip{G}{\Co{G,A}}{\lt \otimes \alpha} $-covariant modular representation, so $ \overline{\Pi}_{\X,\pi^{\X,\rho,S}_{\nu},R} $ is a non-degenerate $ \ast $-representation of $ \Cstar{G,\Co{G,A},\lt \otimes \alpha} $ on $ \X $. However, \autoref{The Covariant Pre-Stone-von Neumann Theorem} says that
$$
\overline{\Pi}_{\LL{G}{A}{\alpha},\pi^{\LL{G}{A}{\alpha},\M,\V}_{\nu},\U}: \Cstar{G,\Co{G,A},\lt \otimes \alpha} \to \Comp{\LL{G}{A}{\alpha}}
$$
is a $ C^{\ast} $-isomorphism, so it follows from \autoref{A Decomposition Result for Non-Degenerate *-Representations of K(X)} that
$$
         \Pair{\X}
              {
              \overline{\Pi}_{\X,\pi^{\X,\rho,S}_{\nu},R} \circ
              \Br{\overline{\Pi}_{\LL{G}{A}{\alpha},\pi^{\LL{G}{A}{\alpha},\M,\V}_{\nu},\U}}^{- 1}
              }
\sim_{W} \bigoplus_{i \in I} \Pair{\LL{G}{A}{\alpha}}{i_{\Comp{\LL{G}{A}{\alpha}} \into \Adj{\LL{G}{A}{\alpha}}}}
$$
for some index set $ I $ and some $ \ds W \in \Unitary{\X,\bigoplus_{i \in I} \LL{G}{A}{\alpha}} $. We thus have
$$
  \FUNC{\overline{\Pi}_{\X,\pi^{\X,\rho,S}_{\nu},R} \circ \Br{\overline{\Pi}_{\LL{G}{A}{\alpha},\pi^{\LL{G}{A}{\alpha},\M,\V}_{\nu},\U}}^{- 1}}
       {T}
= W^{\ast} \Br{\oplus_{i \in I} T} W,
$$
for all $ T \in \Comp{\LL{G}{A}{\alpha}} $, or equivalently,
$$
  \func{\overline{\Pi}_{\X,\pi^{\X,\rho,S}_{\nu},R}}{F}
= W^{\ast} \SqBr{\oplus_{i \in i} \func{\overline{\Pi}_{\LL{G}{A}{\alpha},\pi^{\LL{G}{A}{\alpha},\M,\V}_{\nu},\U}}{F}} W
$$
for all $ F \in \Cstar{G,\Co{G,A},\lt \otimes \alpha} $. It follows from \autoref{Recovering a Covariant Modular Representation from Its Integrated Form} that
$$
\func{R}{x} = W^{\ast} \SqBr{\oplus_{i \in I} \func{\U}{x}} W
\qquad \text{and} \qquad
\func{\pi^{\X,\rho,S}_{\nu}}{g} = W^{\ast} \SqBr{\oplus_{i \in I} \func{\pi^{\LL{G}{A}{\alpha},\M,\V}_{\nu}}{g}} W
$$
for all $ x \in G $ and $ g \in \Co{G,A} $. However, as
$$
\pi^{\X,\rho,S}_{\nu} = \overline{\Pi}_{\X,\rho,S} \circ \overline{\F{\nu}{A}}^{- 1}
\qquad \text{and} \qquad
\pi^{\LL{G}{A}{\alpha},\M,\V}_{\nu} = \overline{\Pi}_{\LL{G}{A}{\alpha},\M,\V} \circ \overline{\F{\nu}{A}}^{- 1},
$$
we find that
$$
\func{\overline{\Pi}_{\X,\rho,S}}{f} = W^{\ast} \SqBr{\oplus_{i \in I} \func{\overline{\Pi}_{\LL{G}{A}{\alpha},\M,\V}}{f}} W
$$
for all $ f \in \Cstar{\widehat{G},A,\iota} $. Another application of \autoref{Recovering a Covariant Modular Representation from Its Integrated Form} yields
$$
\func{S}{\varphi} = W^{\ast} \SqBr{\oplus_{i \in I} \func{\V}{\varphi}} W
\qquad \text{and} \qquad
\func{\rho}{a} = W^{\ast} \SqBr{\oplus_{i \in I} \func{\M}{a}} W
$$
for all $ \varphi \in \widehat{G} $ and $ a \in A $. The covariant Stone-von Neumann Theorem is hereby established.
\end{proof}


Our method of proof in no way depended on the classical Stone-von Neumann Theorem, so it is a proper generalization in every way, as expressed by the corollary below.



\begin{Cor} \label{The Classical Stone-von Neumann Theorem}
The classical Stone-von Neumann Theorem is precisely the case when $ \H = \C $ (any strongly-continuous action of a locally compact Hausdorff group on $ \C $ is necessarily trivial).
\end{Cor}




\section{The Non-Triviality of the Covariant Stone-von Neumann Theorem}


One may now ask, ``Does the covariant Stone-von Neumann Theorem really say anything new? Is there a unitary transformation that reduces it to the case of the trivial action of $ G $ on $ \Comp{\H} $?'' The following result makes this question an extremely valid one.



\begin{Prop} \label{Untwisting L2(G,A,alpha)}
Let $ \Trip{G}{A}{\alpha} $ be a $ C^{\ast} $-dynamical system with $ G $ not assumed to be abelian. Then there is a Hilbert $ A $-module isomorphism $ \Omega: \LL{G}{A}{\alpha} \to \LL{G}{A}{\iota} $ that satisfies
$$
\forall \phi \in \Cc{G,A}: \qquad
\func{\Omega}{\func{q_{\Trip{G}{A}{\alpha}}}{\phi}} = \func{q_{\Trip{G}{A}{\iota}}}{\Map{G}{A}{x}{\alphArg{x^{- 1}}{\func{f}{x}}}}.
$$
\end{Prop}

\begin{proof}
This is an easy verification that we leave to the reader.
\end{proof}


Even though $ \LL{G}{A}{\alpha} $ is isomorphic to $ \LL{G}{A}{\iota} $, note that the covariant Stone-von Neumann Theorem is not a statement about the unitary equivalence of Hilbert $ C^{\ast} $-modules, but a statement about the unitary equivalence of Heisenberg modular representations. Having said this, the next two results give a complete answer to the question above.



\begin{Prop} \label{A Non-Unitary-Equivalence Result}
Let $ \Trip{G}{A}{\alpha} $ and $ \Trip{G}{A}{\beta} $ be $ C^{\ast} $-dynamical systems, with $ G $ abelian and $ \alpha \neq \beta $. Then a direct sum of copies of the $ \Trip{G}{A}{\alpha} $-Schr\"odinger modular representation cannot be unitarily equivalent to a direct sum of copies of the $ \Trip{G}{A}{\beta} $-Schr\"odinger modular representation.
\end{Prop}

\begin{proof}
By way of contradiction, suppose that there are index sets $ I $ and $ J $ such that
$$
         \bigoplus_{i \in I} \Quad{\LL{G}{A}{\alpha}}{\M^{\Trip{G}{A}{\alpha}}}{\U^{\Trip{G}{A}{\alpha}}}{\V^{\Trip{G}{A}{\alpha}}}
\sim_{W} \bigoplus_{j \in J} \Quad{\LL{G}{A}{\beta}}{\M^{\Trip{G}{A}{\beta}}}{\U^{\Trip{G}{A}{\beta}}}{\V^{\Trip{G}{A}{\beta}}}
$$
for some $ \ds W \in \Unitary{\bigoplus_{i \in i} \LL{G}{A}{\alpha},\bigoplus_{j \in J} \LL{G}{A}{\beta}} $. Then we have for all $ x \in G $ and $ a \in A $ that
\begin{align*}
    \func{\U^{\Trip{G}{A}{\alpha}}}{x} \func{\M^{\Trip{G}{A}{\alpha}}}{a}
& = \func{\M^{\Trip{G}{A}{\alpha}}}{\alphArg{x}{a}} \func{\U^{\Trip{G}{A}{\alpha}}}{x}, \\
    \func{\U^{\Trip{G}{A}{\beta}}}{x} \func{\M^{\Trip{G}{A}{\beta}}}{a}
& = \func{\M^{\Trip{G}{A}{\beta}}}{\betArg{x}{a}} \func{\U^{\Trip{G}{A}{\beta}}}{x}, \\
    W \SqBr{\oplus_{i \in I} \func{\U^{\Trip{G}{A}{\alpha}}}{x}} W^{\ast}
& = \oplus_{j \in J} \func{\U^{\Trip{G}{A}{\beta}}}{x}, \\
    W \SqBr{\oplus_{i \in I} \func{\M^{\Trip{G}{A}{\alpha}}}{a}} W^{\ast}
& = \oplus_{j \in J} \func{\M^{\Trip{G}{A}{\beta}}}{a},
\end{align*}
so it follows that
\begin{align*}
    \oplus_{j \in J} \func{\M^{\Trip{G}{A}{\beta}}}{\betArg{x}{a}}
& = \oplus_{j \in J} \func{\U^{\Trip{G}{A}{\beta}}}{x} \func{\M^{\Trip{G}{A}{\beta}}}{a} \func{\U^{\Trip{G}{A}{\beta}}}{x}^{- 1} \\
& = \SqBr{\oplus_{j \in J} \func{\U^{\Trip{G}{A}{\beta}}}{x}}
    \SqBr{\oplus_{j \in J} \func{\M^{\Trip{G}{A}{\beta}}}{a}}
    \SqBr{\oplus_{j \in J} \func{\U^{\Trip{G}{A}{\beta}}}{x}^{- 1}} \\
& = W
    \SqBr{\oplus_{i \in I} \func{\U^{\Trip{G}{A}{\alpha}}}{x}}
    \SqBr{\oplus_{i \in I} \func{\M^{\Trip{G}{A}{\alpha}}}{a}}
    \SqBr{\oplus_{i \in I} \func{\U^{\Trip{G}{A}{\alpha}}}{x}^{- 1}}
    W^{\ast} \\
& = W \SqBr{\oplus_{i \in I}
    \func{\U^{\Trip{G}{A}{\alpha}}}{x}
    \func{\M^{\Trip{G}{A}{\alpha}}}{a}
    \func{\U^{\Trip{G}{A}{\alpha}}}{x}^{- 1}}
    W^{\ast} \\
& = W \SqBr{\oplus_{i \in I} \func{\M^{\Trip{G}{A}{\alpha}}}{\alphArg{x}{a}}} W^{\ast} \\
& = \oplus_{j \in J} \func{\M^{\Trip{G}{A}{\beta}}}{\alphArg{x}{a}},
\end{align*}
which yields $ \func{\M^{\Trip{G}{A}{\beta}}}{\betArg{x}{a}} = \func{\M^{\Trip{G}{A}{\beta}}}{\alphArg{x}{a}} $. Hence, for all $ x \in G $, $ a \in A $, and $ \phi \in \Cc{G,A} $,
\begin{align*}
    \func{q_{\Trip{G}{A}{\beta}}}{\Map{G}{A}{y}{\alphArg{x}{a} \func{\phi}{y}}}
& = \FUNC{\func{\M^{\Trip{G}{A}{\beta}}}{\alphArg{x}{a}}}{\func{q_{\Trip{G}{A}{\beta}}}{\phi}} \\
& = \FUNC{\func{\M^{\Trip{G}{A}{\beta}}}{\betArg{x}{a}}}{\func{q_{\Trip{G}{A}{\beta}}}{\phi}} \\
& = \func{q_{\Trip{G}{A}{\beta}}}{\Map{G}{A}{y}{\betArg{x}{a} \func{\phi}{y}}},
\end{align*}
from which we get $ \Br{\alphArg{x}{a} - \betArg{x}{a}} \func{\phi}{y} = 0_{A} $ for all $ y \in G $. As we can choose $ \phi $ to assume any value at any point, we obtain $ \alphArg{x}{a} = \betArg{x}{a} $ for all $ x \in G $ and $ a \in A $, which contradicts $ \alpha \neq \beta $.
\end{proof}



\begin{Cor} \label{The Non-Triviality of the Covariant Stone-von Neumann Theorem}
Let $ \Trip{G}{\Comp{\H}}{\alpha} $ and $ \Trip{G}{\Comp{\H}}{\beta} $ be $ C^{\ast} $-dynamical systems, with $ G $ abelian, $ \H $ a non-trivial Hilbert space, and $ \alpha \neq \beta $. Then any $ \Trip{G}{\Comp{\H}}{\alpha} $-Heisenberg modular representation cannot be unitarily equivalent to any $ \Trip{G}{\Comp{\H}}{\beta} $-Heisenberg modular representation.
\end{Cor}

\begin{proof}
This follows immediately from \autoref{The Covariant Stone-von Neumann Theorem} and \autoref{A Non-Unitary-Equivalence Result}.
\end{proof}


\autoref{The Non-Triviality of the Covariant Stone-von Neumann Theorem} should remind physicists of Haag's Theorem in quantum field theory (QFT), which posits the failure of the uniqueness of the canonical commutation relations within QFT in general (\cite{Ha}).

We finally arrive at a discussion of Takai-Takesaki Duality.



\begin{Thm}[Takai-Takesaki Duality \cite{Ra,Wi}] \label{Takai-Takesaki Duality}
Let $ \Trip{G}{A}{\alpha} $ be a $ C^{\ast} $-dynamical system. Then
$$
\Cstar{\widehat{G},\Cstar{G,A,\alpha},\hat{\alpha}} \cong \Comp{\L{2}{G}} \otimes A,
$$
where $ \hat{\alpha} $ denotes the dual action of $ \widehat{G} $ on $ \Cstar{G,A,\alpha} $.
\end{Thm}


In his proof of Takai-Takesaki Duality in \cite{Ra}, Iain Raeburn first showed that
$$
\Cstar{\widehat{G},\Cstar{G,A,\alpha},\hat{\alpha}} \cong \Cstar{G,\Co{G,A},\lt \otimes \alpha}.
$$
He then formed a $ C^{\ast} $-isomorphism $ \Cstar{\widehat{G},\Cstar{G,A,\alpha},\hat{\alpha}} \cong \Comp{\L{2}{G}} \otimes A $ as a composition of a series of $ C^{\ast} $-isomorphisms shown below, each requiring a lengthy justification except for the last one:
$$
      \Cstar{G,\Co{G,A},\lt \otimes \alpha}
\cong \Cstar{G,\Co{G,A},\lt \otimes \iota}
\cong \Cstar{G,\Co{G},\lt} \otimes A
\cong \Comp{\L{2}{G}} \otimes A.
$$
His ``untwisting'' of $ \alpha $ is thus performed at the level of $ C^{\ast} $-crossed products, with the last $ C^{\ast} $-isomorphism being given by the classical Stone-von Neumann Theorem, which relies on Green's Imprimitivity Theorem. However, by taking full advantage of Green's Imprimitivity Theorem, we can derive a shorter proof of this $ C^{\ast} $-isomorphism, which ``untwists'' $ \alpha $ at the level of Hilbert $ C^{\ast} $-modules:



\begin{Prop} \label{The Second Half of Raeburn's Proof of Takai-Takesaki Duality}
Let $ \Trip{G}{A}{\alpha} $ be a $ C^{\ast} $-dynamical system. Then $ \Cstar{G,\Co{G,A},\lt \otimes \alpha} \cong \Comp{\L{2}{G}} \otimes A $.
\end{Prop}

\begin{proof}
Using \autoref{The Covariant Pre-Stone-von Neumann Theorem}, \autoref{Untwisting L2(G,A,alpha)}, and basic results about Hilbert $ C^{\ast} $-modules, we offer a one-line proof:
$$
      \Cstar{G,\Co{G,A},\lt \otimes \alpha}
\cong \Comp{\LL{G}{A}{\alpha}}
\cong \Comp{\LL{G}{A}{\iota}}
\cong \Comp{\L{2}{G} \otimes A_{A}}
\cong \Comp{\L{2}{G}} \otimes A. \qedhere
$$
\end{proof}




\section{Conclusions}


We would like to present here some questions and thoughts that naturally arose while writing this paper:
\begin{enumerate}
\item
Is there a $ C^{\ast} $-algebra $ A $ not $ C^{\ast} $-isomorphic to $ \Comp{\H} $ for a Hilbert space $ \H $ such that any $ C^{\ast} $-dynamical system of the form $ \Trip{G}{A}{\alpha} $ has the von Neumann Uniqueness Property? As $ C^{\ast} $-subalgebras of $ \Comp{\H} $ are $ C^{\ast} $-isomorphic to a direct sum $ \ds \bigoplus_{i \in I} \Comp{\H_{i}} $, where the $ \H_{i} $'s are Hilbert spaces, we think that a series of technical extensions can be made to accommodate the Covariant Stone-von Neumann Theorem for such $ C^{\ast} $-algebras.

\item
The results of this paper suggest that quantum mechanics could be developed using Hilbert $ C^{\ast} $-modules as state spaces, in which case the expectations of observables would assume values in a $ C^{\ast} $-algebra. Can this idea be developed further?

\item
As mentioned in the introduction, we suspect that the covariant Stone-von Neumann Theorem could be generalized to actions of non-abelian groups using techniques of non-abelian duality.
\end{enumerate}

While interesting in a purely-mathematical context, the Covariant Stone-von Neumann Theorem has a rich interpretation from the perspective of quantum mechanics. By including representations of $ C^{\ast} $-dynamical systems, it allows for the consideration of time-dependence of observables in addition to time-dependence of states. To contrast, recall that a time-independent quantum system is modeled by a Hilbert space $ \H $ and a Hamiltonian $ \hat{H} $ whose corresponding one-parameter unitary family, $ \Seq{e^{- \Br{i t / \hbar} \cdot \hat{H}}}{t \in \R} $, determines the time evolution of the state space via
$$
\forall \psi \in \H, ~ \forall t \in \R: \qquad
\func{\psi}{t} = e^{- \Br{i t / \hbar} \cdot \hat{H}} \cdot \func{\psi}{0}.
$$

The time evolution of the state space determined by $ \hat{H} $ can also be viewed as time evolution of the algebra $ \Bdd{\H} $ of bounded observables via $ \Map{\R}{\Bdd{\H}}{t}{e^{\Br{i t / \hbar} \cdot \hat{H}} T e^{- \Br{i t / \hbar} \cdot \hat{H}}} $, for all $ T \in \Bdd{\H} $. From this perspective, one may state the time-independent version of Ehrenfest's Theorem:
$$
\frac{\d}{\d t} \Inner{\psi}{\func{T}{\psi}}_{\H} = \Inner{\psi}{\func{\Comm{i \cdot \hat{H}}{T}}{\psi}}_{\H}.
$$

As the Covariant Stone-von Neumann Theorem applies to $ C^{\ast} $-dynamical systems of the form $ \Trip{G}{\Comp{\H}}{\alpha} $, and as all $ \ast $-automorphisms of $ \Comp{\H} $ are implemented via conjugation by unitaries, we make a convenient but natural restriction in the case when $ G = \R $ to the action $ \alpha^{C} $, where $ C \in \Bdd{\H} $ is self-adjoint, and
$$
\forall T \in \Bdd{\H}, ~ \forall t \in \R: \qquad
\func{\alpha^{C}_{t}}{T} \df e^{\Br{i t / \hbar} \cdot C} T e^{- \Br{i t / \hbar} \cdot C}.
$$
The covariance conditions present in the definition of an $ \Trip{\R}{\Comp{\H}}{\alpha^{C}} $-Heisenberg modular representation $ \Quad{\X}{\rho}{R}{S} $ then reduce to commutation relations between $ C $ and the infinitesimal generators of $ R $ and $ S $. It is in this context that we are able to get an infinitesimal version of the Covariant Stone-von Neumann Theorem, which will appear in a sequel to this article.

As mentioned in the introduction, a catalyst for the Stone-von Neumann Theorem was to investigate the uniqueness of pairs $ \Pair{A}{B} $ of self-adjoint Hilbert-space operators satisfying the Heisenberg Commutation Relation. Nelson's counterexample \cite{Ne} shows that uniqueness fails in general, and decades of research have been devoted to identifying sufficient conditions for $ \Pair{A}{B} $ that imply that $ \Seq{e^{i s \cdot A}}{s \in \R} $ and $ \Seq{e^{i t \cdot B}}{t \in \R} $ satisfy the Weyl Commutation Relation.

In the sequel, we follow the strategy in \cite{Hu} --- which takes place in the Hilbert-space setting --- to provide necessary and sufficient conditions for when a pair $ \Pair{A}{B} $ of unbounded self-adjoint operators on a Hilbert $ \Comp{\H} $-module yield one-parameter unitary groups that satisfy the Weyl Commutation Relation.




\section*{Appendix}


\begin{proof}[Proof of \autoref{Recovering a Covariant Modular Representation from Its Integrated Form}]
Let $ \mathcal{N} $ and $ \mathcal{O} $ be neighborhood bases of $ x $ and $ e_{G} $ in $ G $ respectively, directed by reverse inclusion. By Urysohn's Lemma, we can find nets $ \Seq{\phi_{U}}{U \in \mathcal{N}} $ and $ \Seq{\psi_{V}}{V \in \mathcal{O}} $ in $ \Cc{G,\R_{\geq 0}} $ with the following properties:
\begin{itemize}
\item
$ \Supp{\phi_{U}} \subseteq U $ for each $ U \in \mathcal{N} $, and $ \Supp{\psi_{V}} \subseteq V $ for each $ V \in \mathcal{O} $.

\item
$ \displaystyle \Int{G}{\func{\phi_{U}}{y}}{\mu}{y} = 1 = \Int{G}{\func{\psi_{V}}{y}}{\mu}{y} $ for all $ U \in \mathcal{N} $ and $ V \in \mathcal{O} $.
\end{itemize}
Also, let $ \Seq{e_{\lambda}}{\lambda \in \Lambda} $ be an approximate identity for $ A $ norm-bounded by $ 1 $.

Let $ \zeta \in \X $ and $ \epsilon > 0 $. As $ \func{R}{y} $ converges strongly in $ \Adj{\X} $ to $ \func{R}{x} $ as $ y \to x $, we can find a $ U_{0} \in \mathcal{N} $ such that for all $ y \in U_{0} $,
$$
\Norm{\FUNC{\func{R}{x}}{\zeta} - \FUNC{\func{R}{y}}{\zeta}}_{\X} < \frac{\epsilon}{3}.
$$
As $ \rho $ is non-degenerate, $ \Seq{\func{\rho}{e_{\lambda}}}{\lambda \in \Lambda} $ converges strongly in $ \Adj{\X} $ to $ \Id_{\X} $, so we can find a $ \lambda_{0} \in \Lambda $ such that
$$
\forall \lambda \in \Lambda_{\geq \lambda_{0}}: \qquad
  \Norm{\FUNC{\func{R}{x}}{\zeta} - \FUNC{\func{\rho}{e_{\lambda}} \circ \func{R}{x}}{\zeta}}_{\X}
= \Norm{\FUNC{\func{R}{x}}{\zeta} - \FUNC{\func{\rho}{e_{\lambda}}}{\FUNC{\func{R}{x}}{\zeta}}}_{\X}
< \frac{\epsilon}{3}.
$$
We then have for all $ y \in U_{0} $ and $ \lambda \in \Lambda_{\geq \lambda_{0}} $ that
\begin{align*}
     & ~ \Norm{\FUNC{\func{R}{x}}{\zeta} - \FUNC{\func{\rho}{e_{\lambda}} \circ \func{R}{y}}{\zeta}}_{\X} \\
\leq & ~ \Norm{\FUNC{\func{R}{x}}{\zeta} - \FUNC{\func{\rho}{e_{\lambda}} \circ \func{R}{x}}{\zeta}}_{\X} +
         \Norm{\FUNC{\func{\rho}{e_{\lambda}} \circ \func{R}{x}}{\zeta} - \FUNC{\func{\rho}{e_{\lambda}} \circ \func{R}{y}}{\zeta}}_{\X} \\
=    & ~ \Norm{\FUNC{\func{R}{x}}{\zeta} - \FUNC{\func{\rho}{e_{\lambda}} \circ \func{R}{x}}{\zeta}}_{\X} +
         \Norm{\FUNC{\func{\rho}{e_{\lambda}}}{\FUNC{\func{R}{x}}{\zeta} - \FUNC{\func{R}{y}}{\zeta}}}_{\X} \\
\leq & ~ \Norm{\FUNC{\func{R}{x}}{\zeta} - \FUNC{\func{\rho}{e_{\lambda}} \circ \func{R}{x}}{\zeta}}_{\X} +
         \Norm{\func{\rho}{e_{\lambda}}}_{\Adj{\X}} \Norm{\FUNC{\func{R}{x}}{\zeta} - \FUNC{\func{R}{y}}{\zeta}}_{\X} \\
\leq & ~ \Norm{\FUNC{\func{R}{x}}{\zeta} - \FUNC{\func{\rho}{e_{\lambda}} \circ \func{R}{x}}{\zeta}}_{\X} +
         \Norm{e_{\lambda}}_{A} \Norm{\FUNC{\func{R}{x}}{\zeta} - \FUNC{\func{R}{y}}{\zeta}}_{\X} \\
\leq & ~ \Norm{\FUNC{\func{R}{x}}{\zeta} - \FUNC{\func{\rho}{e_{\lambda}} \circ \func{R}{x}}{\zeta}}_{\X} +
         \Norm{\FUNC{\func{R}{x}}{\zeta} - \FUNC{\func{R}{y}}{\zeta}}_{\X} \\
<    & ~ \frac{2 \epsilon}{3}.
\end{align*}
It thus follows for all $ U \in \mathcal{N}_{\subseteq U_{0}} $ and $ \lambda \in \Lambda_{\geq \lambda_{0}} $ that
\begin{align*}
     & ~ \Norm{\FUNC{\func{R}{x}}{\zeta} - \FUNC{\func{\Pi_{\X,\rho,R}}{\phi_{U} \diamond e_{\lambda}}}{\zeta}}_{\X} \\
=    & ~ \Norm{
              \FUNC{\func{R}{x}}{\zeta} - \FUNC{\Int{G}{\func{\rho}{\Func{\phi_{U} \diamond e_{\lambda}}{y}} \circ \func{R}{y}}{\mu}{y}}{\zeta}
              }_{\X} \\
=    & ~ \Norm{
              \FUNC{\func{R}{x}}{\zeta} - \Int{G}{\func{\phi_{U}}{y} \cdot \FUNC{\func{\rho}{e_{\lambda}} \circ \func{R}{y}}{\zeta}}{\mu}{y}
              }_{\X} \\
=    & ~ \Norm{
              \underbrace{\SqBr{\Int{G}{\func{\phi_{U}}{y}}{\mu}{y}}}_{= 1} \cdot \FUNC{\func{R}{x}}{\zeta} -
              \Int{G}{\func{\phi_{U}}{y} \cdot \FUNC{\func{\rho}{e_{\lambda}} \circ \func{R}{y}}{\zeta}}{\mu}{y}
              }_{\X} \\
=    & ~ \Norm{
              \Int{G}{\func{\phi_{U}}{y} \cdot \SqBr{\FUNC{\func{R}{x}}{\zeta} - \FUNC{\func{\rho}{e_{\lambda}} \circ \func{R}{y}}{\zeta}}}
                  {\mu}{y}
              }_{\X} \\
\leq & ~ \Int{G}
             {\Norm{\func{\phi_{U}}{y} \cdot \SqBr{\FUNC{\func{R}{x}}{\zeta} - \FUNC{\func{\rho}{e_{\lambda}} \circ \func{R}{y}}{\zeta}}}_{\X}}
             {\mu}{y} \\
=    & ~ \Int{G}{\func{\phi_{U}}{y} \Norm{\FUNC{\func{R}{x}}{\zeta} - \FUNC{\func{\rho}{e_{\lambda}} \circ \func{R}{y}}{\zeta}}_{\X}}{\mu}{y}
         \\
=    & ~ \Int{U}{\func{\phi_{U}}{y} \Norm{\FUNC{\func{R}{x}}{\zeta} - \FUNC{\func{\rho}{e_{\lambda}} \circ \func{R}{y}}{\zeta}}_{\X}}{\mu}{y}
         \\
\leq & ~ \Int{U}{\func{\phi_{U}}{y} \cdot \frac{2 \epsilon}{3}}{\mu}{y} \\
=    & ~ \frac{2 \epsilon}{3} \underbrace{\Int{U}{\func{\phi_{U}}{y}}{\mu}{y}}_{= 1} \\
<    & ~ \epsilon.
\end{align*}
As $ \zeta \in \X $ and $ \epsilon > 0 $ are arbitrary, the net $ \Seq{\func{\Pi_{\X,\rho,R}}{\phi_{U} \diamond e_{\lambda}}}{\Pair{U}{\lambda} \in \mathcal{N} \times \Lambda} $ converges strongly in $ \Adj{\X} $ to $ \func{R}{x} $.

Let $ \zeta \in \X $ and $ \epsilon > 0 $. As $ \func{R}{y} $ converges strongly in $ \Adj{\X} $ to $ \func{R}{e_{H}} = \Id_{\X} $ as $ y \to e_{H} $, we can find a $ V_{0} \in \mathcal{O} $ such that
$$
\forall y \in V_{0}: \qquad
\Norm{\zeta - \FUNC{\func{R}{y}}{\zeta}}_{\X} < \dfrac{\epsilon}{1 + \Norm{a}_{A}},
$$
so
\begin{align*}
\forall y \in V_{0}: \qquad
       \Norm{\FUNC{\func{\rho}{a}}{\zeta} - \FUNC{\func{\rho}{a} \circ \func{R}{y}}{\zeta}}_{\X}
& =    \Norm{\FUNC{\func{\rho}{a}}{\zeta - \FUNC{\func{R}{y}}{\zeta}}}_{\X} \\
& \leq \Norm{\func{\rho}{a}}_{\Adj{\X}} \Norm{\zeta - \FUNC{\func{R}{y}}{\zeta}}_{\X} \\
& \leq \Norm{a}_{A} \Norm{\zeta - \FUNC{\func{R}{y}}{\zeta}}_{\X} \\
& \leq \Norm{a}_{a} \cdot \dfrac{\epsilon}{1 + \Norm{a}_{A}}.
\end{align*}
It follows for all $ V \in \mathcal{O}_{\subseteq V_{0}} $ that
\begin{align*}
     & ~ \Norm{\FUNC{\func{\rho}{a}}{\zeta} - \FUNC{\func{\Pi_{\X,\rho,R}}{\psi_{V} \diamond a}}{\zeta}}_{\X} \\
=    & ~ \Norm{\FUNC{\func{\rho}{a}}{\zeta} - \FUNC{\Int{G}{\func{\rho}{\Func{\psi_{V} \diamond a}{y}} \circ \func{R}{y}}{\mu}{y}}{\zeta}}_{\X}
         \\
=    & ~ \Norm{\FUNC{\func{\rho}{a}}{\zeta} - \Int{G}{\func{\psi_{V}}{y} \cdot \FUNC{\func{\rho}{a} \circ \func{R}{y}}{\zeta}}{\mu}{y}}_{\X} \\
=    & ~ \Norm{
              \underbrace{\SqBr{\Int{G}{\func{\psi_{V}}{y}}{\mu}{y}}}_{= 1} \cdot \FUNC{\func{\rho}{a}}{\zeta} -
              \Int{G}{\func{\psi_{V}}{y} \cdot \FUNC{\func{\rho}{a} \circ \func{R}{y}}{\zeta}}{\mu}{y}
              }_{\X} \\
=    & ~ \Norm{
              \Int{G}{\func{\psi_{V}}{y} \cdot \SqBr{\FUNC{\func{\rho}{a}}{\zeta} - \FUNC{\func{\rho}{a} \circ \func{R}{y}}{\zeta}}}{\mu}{y}
              }_{\X} \\
\leq & ~ \Int{G}{\Norm{\func{\psi_{V}}{y} \cdot \SqBr{\FUNC{\func{\rho}{a}}{\zeta} - \FUNC{\func{\rho}{a} \circ \func{R}{y}}{\zeta}}}_{\X}}
             {\mu}{y} \\
=    & ~ \Int{G}{\func{\psi_{V}}{y} \Norm{\FUNC{\func{\rho}{a}}{\zeta} - \FUNC{\func{\rho}{a} \circ \func{R}{y}}{\zeta}}_{\X}}{\mu}{y} \\
=    & ~ \Int{V}{\func{\psi_{V}}{y} \Norm{\FUNC{\func{\rho}{a}}{\zeta} - \FUNC{\func{\rho}{a} \circ \func{R}{y}}{\zeta}}_{\X}}{\mu}{y} \\
\leq & ~ \Int{V}{\func{\psi_{V}}{y} \Br{\Norm{a}_{A} \cdot \frac{\epsilon}{1 + \Norm{a}_{A}}}}{\mu}{y} \\
=    & ~ \Br{\Norm{a}_{A} \cdot \frac{\epsilon}{1 + \Norm{a}_{A}}} \underbrace{\Int{V}{\func{\psi_{V}}{y}}{\mu}{y}}_{= 1} \\
<    & ~ \epsilon.
\end{align*}
As $ \zeta \in \X $ and $ \epsilon > 0 $ are arbitrary, the net $ \Seq{\func{\Pi_{\X,\rho,R}}{\psi_{V} \diamond a}}{V \in \mathcal{O}} $ converges strongly in $ \Adj{\X} $ to $ \func{\rho}{a} $.
\end{proof}


\begin{proof}[Proof of \autoref{The Approximation Lemma}]
Fix $ f \in \Cc{X,V} $ and $ \epsilon > 0 $. Let $ K \df \Supp{f} $, and let $ \Seq{U_{p}}{p \in K} $ be a $ K $-indexed sequence of open subsets of $ X $ with the following properties:
\begin{itemize}
\item
$ U_{p} $ is an open neighborhood of $ p $ in $ X $ for each $ p \in K $.

\item
$ \Norm{\func{f}{x} - \func{f}{p}}_{V} < \dfrac{\epsilon}{3} $ for all $ x \in U_{p} $.
\end{itemize}
Clearly, $ \SSet{U_{p}}_{p \in K} $ covers $ K $, and as $ K $ is compact, there is a finite subset $ F $ of $ K $ such that $ \SSet{U_{p}}_{p \in F} $ also covers $ K $, and $ U_{p} \neq U_{p'} $ for distinct $ p,p' \in F $. As $ X $ is locally compact and Hausdorff, there is a partition of unity $ \Seq{\varphi_{p}}{p \in F} $ for $ K $ that is subordinate to $ \SSet{U_{p}}_{p \in F} $, i.e.,
\begin{itemize}
\item
$ \varphi_{p} \in \Cc{X,\SqBr{0,1}} $ and $ \Supp{\varphi_{p}} \subseteq U_{p} $ for each $ p \in F $, and

\item
$ \displaystyle \sum_{p \in F} \func{\varphi_{p}}{x} \leq 1 $ for all $ x \in X $, with equality holding for all $ x \in K $.
\end{itemize}
Define $ P \in \Cc{X,V} $ by
$$
\forall x \in X: \qquad
\func{P}{x} \df \sum_{p \in F} \func{\varphi_{p}}{x} \cdot \func{f}{p}.
$$
For all $ x \in X $, we have $ \displaystyle \func{f}{x} = \sum_{p \in F} \func{\varphi_{p}}{x} \cdot \func{f}{x} $ (if $ x \in K $, then $ \displaystyle \sum_{p \in F} \func{\varphi_{p}}{x} = 1 $; otherwise, $ \func{f}{x} = 0_{V} $), so
\begin{align*}
       \Norm{\func{f}{x} - \func{P}{x}}_{V}
& =    \Norm{\sum_{p \in F} \func{\varphi_{p}}{x} \cdot \func{f}{x} - \sum_{p \in F} \func{\varphi_{p}}{x} \cdot \func{f}{p}}_{V} \\
& =    \Norm{\sum_{p \in F} \func{\varphi_{p}}{x} \cdot \SqBr{\func{f}{x} - \func{f}{p}}}_{V} \\
& \leq \sum_{p \in F} \Norm{\func{\varphi_{p}}{x} \cdot \SqBr{\func{f}{x} - \func{f}{p}}}_{V} \\
& =    \sum_{p \in F} \func{\varphi_{p}}{x} \Norm{\func{f}{x} - \func{f}{p}}_{V} \\
& \leq \sum_{p \in F} \func{\varphi_{p}}{x} \cdot \dfrac{\epsilon}{3} \\
& \leq \dfrac{\epsilon}{3}.
\end{align*}
As $ D $ is a dense subset of $ V $, we can find an $ F $-indexed sequence $ \Seq{v_{p}}{p \in F} $ in $ D $ such that $ \Norm{\func{f}{p} - v_{p}}_{V} < \dfrac{\epsilon}{3} $ for each $ p \in F $. Define $ Q \in \Cc{X,V} $ by
$$
\forall x \in X: \qquad
\func{Q}{x} \df \sum_{p \in F} \func{\varphi_{p}}{x} \cdot v_{p}.
$$
Then
\begin{align*}
\forall x \in X: \qquad
       \Norm{\func{P}{x} - \func{Q}{x}}_{V}
& =    \Norm{\sum_{p \in F} \func{\varphi_{p}}{x} \cdot \func{f}{p} - \sum_{p \in F} \func{\varphi_{p}}{x} \cdot v_{p}}_{V} \\
& =    \Norm{\sum_{p \in F} \func{\varphi_{p}}{x} \cdot \SqBr{\func{f}{p} - v_{p}}}_{V} \\
& \leq \sum_{p \in F} \Norm{\func{\varphi_{p}}{x} \cdot \SqBr{\func{f}{p} - v_{p}}}_{V} \\
& =    \sum_{p \in F} \func{\varphi_{p}}{x} \Norm{\func{f}{p} - v_{p}}_{V} \\
& \leq \sum_{p \in F} \func{\varphi_{p}}{x} \cdot \frac{\epsilon}{3} \\
& \leq \frac{\epsilon}{3}.
\end{align*}
Therefore, by the Triangle Inequality,
$$
\forall x \in X: \qquad
     \Norm{\func{f}{x} - \func{Q}{x}}_{V}
\leq \Norm{\func{f}{x} - \func{P}{x}}_{V} + \Norm{\func{P}{x} - \func{Q}{x}}_{V}
\leq \frac{\epsilon}{3} + \frac{\epsilon}{3}
=    \frac{2 \epsilon}{3}
<    \epsilon.
$$
As $ Q $ has the desired form, the first part of the theorem is therefore established.

Let $ U $ be an open neighborhood of $ K $ whose closure is compact (such a neighborhood exists because $ X $ is locally compact and Hausdorff). Then $ U $ is a locally compact Hausdorff space and $ f|_{U} \in \Cc{U,V} $, so we may apply the first part of the theorem to find $ \varphi_{1},\ldots,\varphi_{n} \in \Cc{U} $ and $ v_{1},\ldots,v_{n} \in D $ such that
$$
\forall x \in U: \qquad
  \Norm{\func{f|_{U}}{x} - \sum_{i = 1}^{n} \func{\varphi_{i}}{x} \cdot v_{i}}_{V}
< \frac{\epsilon}{1 + \func{\nu}{U}}.\footnote{The measure of a pre-compact subset, with respect to a regular Borel measure, is finite.}
$$
Let $ \widetilde{\varphi_{1}},\ldots,\widetilde{\varphi_{n}} $ denote the respective extensions of $ \varphi_{1},\ldots,\varphi_{n} $ to $ X $ by $ 0_{V} $. Then $ \widetilde{\varphi_{1}},\ldots,\widetilde{\varphi_{n}} \in \Cc{X} $, and
\begin{align*}
       \Int{X}{\Norm{\func{f}{x} - \sum_{i = 1}^{n} \func{\widetilde{\varphi_{i}}}{x} \cdot v_{i}}_{V}}{\nu}{x}
& =    \Int{U}{\Norm{\func{f|_{U}}{x} - \sum_{i = 1}^{n} \func{\varphi_{i}}{x} \cdot v_{i}}_{V}}{\nu}{x} \\
& \leq \Int{U}{\frac{\epsilon}{1 + \func{\nu}{U}}}{\nu}{x} \\
& =    \func{\nu}{U} \cdot \frac{\epsilon}{1 + \func{\nu}{U}} \\
& <    \epsilon.
\end{align*}
The proof of the second part of the theorem is therefore established.
\end{proof}




\section*{Acknowledgments}


The first author wishes to thank the second author for her warm hospitality and collaborative energy when he visited the University of Nebraska-Lincoln. He would also like to thank his former PhD advisor, Professor Albert Sheu of the University of Kansas, for cultivating his deep-seated interest in the Stone-von Neumann Theorem in graduate school.

The second author would like to thank the first author for his invitation to collaborate on this project. His passion and breadth of knowledge have been a privilege to work alongside. The second author is also grateful to her advisors, Allan Donsig and David Pitts, as well as Ruy Exel, for their thoughtful comments and questions.





\end{document}